%% file: paper3.tex
\documentclass[11pt,a4paper,reqno]{article}
\usepackage[margin=1in]{geometry}

\usepackage{acro}
\input{acronyms.tex}

\usepackage{amssymb}
\usepackage{latexsym}
\usepackage{amsmath}
\usepackage{graphicx}
\usepackage{amsthm}
\usepackage{empheq}
\usepackage{bm}
\usepackage{booktabs}
\usepackage[dvipsnames,table]{xcolor}
\usepackage{pagecolor}
\usepackage{caption}
\usepackage{subcaption}
\usepackage{tikz,pgfplots}
\usetikzlibrary{matrix,decorations.pathreplacing,positioning}
\usepackage{enumitem}
\usepackage{wasysym}

\usetikzlibrary{positioning}

\usepackage{amsmath}
\providecommand{\keywords}[1]
{
  \small	
  \textbf{\textit{Keywords---}} #1
}

\definecolor{softblue}{RGB}{0,162,255}

\numberwithin{equation}{section}

\theoremstyle{definition}
\newtheorem{theorem}{Theorem}[section]

\newtheorem{definition}[theorem]{Definition}
\newtheorem{example}[theorem]{Example}

\newtheorem{lemma}[theorem]{Lemma}

\newcommand{\Z}{\mathbb{Z}} 

\newcommand{\R}{\mathbb{R}}

\providecommand{\keywords}[1]{\textbf{\textit{Index terms---}} #1}

\newlength{\mynodespace}
\definecolor{myg}{RGB}{220,220,220}
\usepackage{soul}
\pgfplotsset{compat=1.18}
\usepackage{calrsfs}
\usepackage{float}

\usepackage{xcolor}
\usepackage{soul}

\usepackage{hyperref}
\usepackage{authblk}
\newcommand*\samethanks[1][\value{footnote}]{\footnotemark[#1]}

\title{Quantum Approaches to the Quadratic Assignment Problem}
\author[1]{Nathan Daly\thanks{Funding for this work was provided by the Department of Navy award N00174-22-1-0030 issued by the Office of Naval Research.}}

\author[2]{Thomas Krauss} 

\author[3]{Julia Shapiro\samethanks{}\thanks{J. Shapiro is supported by the Department of Defense Cyber Service Academy Scholarship.}}

\affil[1]{Department of Mathematics, Virginia Tech, Blacksburg, Virginia, U.S.A.}

\affil[2]{I/ONX High Performance Compute, Las Vegas, Nevada, U.S.A.}

\date{}

\begin{document}

\maketitle

\abstract{The Quadratic Assignment Problem (QAP) is an NP-hard fundamental combinatorial optimization problem introduced by Koopmans and Beckmann in 1957. The problem is to assign $n$ facilities to $n$ different locations with the goal of minimizing the cost of the total distances between facilities weighted by the corresponding flows. We initiate the study of using Rydberg arrays to find optimal solutions to the QAP and provide a complementing circuit theory to facilitate an easy representation of other hard problems. We provide an algorithm for finding valid and optimal solutions to the QAP using Rydberg arrays.}


\section{Introduction}
The QAP was introduced by Koopsman and Beckmann in \cite{qap} and is a fundamental problem in combinatorial optimization and operations research. It represents a challenging class of problems related to assigning a set of $n$ facilities to a set of $n$ locations with the goal of minimizing the cost associated with the assignments. The cost is a product of the distances between locations and the pair-wise flow between facilities, making the \ac{QAP} particularly useful in scenarios where both spatial positioning and inter-facility interactions are critical. Originally formulated to address facility layout issues,~\ac{QAP} has since found applications in various fields: for instance, in manufacturing, it optimizes layout to reduce transportation costs or improve process efficiency; in data centers, it helps assign servers to physical locations to minimize latency and energy use. Beyond these, \ac{QAP} has applications in DNA sequencing \cite{array}, airport terminal planning \cite{terminal}, and backboard layout \cite{backboard}. Solving \ac{QAP} is computationally intensive and research in the field focuses on both exact and heuristic approaches. Many research directions for finding optimal solutions to the~\ac{QAP} have focused on classical methods \cite{convexB, surveyMain, QAP1, FRIEZE198389,greedy, algQ, surveryQ, SolvingQA, TATE195AGA, recentQ, 4658351, simulatedA, mip} but these methods are constrained to small problem sizes. Recent work using Rydberg arrays and QuEra's Aquilas hardware, a 256 neutral atom quantum computer, to solve the Maximum Independent Set (MIS) problem yielded promising results \cite{quantumOpt}. In \cite{nguyen2023quantum}, the authors introduce a general framework for
mapping combinatorial optimization problems to unit disk graphs (UDG) and Rydberg arrays such as the MIS problem and the Maximum Weighted Independent Set (MWIS) problem and extend this to general classes of constrained binary optimization problems.

In this paper, we initiate the study of the \ac{QAP} using Rydberg arrays. In particular, we provide a general framework for circuit construction for classes of constrained binary optimization problems and we apply the gadget construction provided in \cite{nguyen2023quantum} to enforce constraints for the \ac{QAP} problem. The framework provided facilitates the design of the reduction from the \ac{QAP} to UDG-MWIS and of quantum algorithms for other problems. We show, through the general framework provided and the reduction of QAP to~UDG-MWIS, that Rydberg arrays can not only be used to solve other hard problems but also be optimized for improved performance.

This paper is organized as follows. In Section~\ref{background}, we will discuss the problem set up and the Rydberg array approach to solving the maximum independent set and maximum weighted independent set problems provided in \cite{nguyen2023quantum}.  In Section~\ref{circuit}, we present our framework for circuit construction of a broader class of constrained binary optimization problems. In Section~\ref{reduction}, we extend the techniques applied to the MIS and MWIS problems in \cite{nguyen2023quantum} and provide a general process for mapping \ac{QAP} to an instance of MWIS and show that this reduction encodes valid and optimal solutions of the \ac{QAP}. We then provide an optimized algorithm that reduces the number of variables needed to represent instance of the \ac{QAP}. We show that the MWIS of the graph constructed in the reduction corresponds to optimal solutions of the original \ac{QAP} instance.

\section{Preliminaries and Prior Work}\label{background} 

In this section, we formally introduce the \ac{QAP} and the approach to solving the \ac{MIS} and \ac{MWIS} on Rydberg neutral atom hardware using the formulation provided in \cite{nguyen2023quantum}.

\subsection{The Quadratic Assignment Problem}

The Quadratic Assignment Problem models the following situation: given $n$ facilities and~$n$ locations, assign all facilities to different locations with the goal of minimizing the cost, that is, the sum of the distances between locations multiplied by the corresponding flows between facilities. 
\noindent The model is based on the following constraints:
\begin{enumerate}
    \item Each facility can only be in one location,
    \item Each location can take only one facility, and,
    \item The cost of this process is measured by the pair-wise distance between locations times flow between facilities at the locations
\end{enumerate}

Formally, an instance of the \ac{QAP} is the tuple $I = (F,D)$, where $F$ is the flow matrix each of whose entries $f_{xy}$ is the flow from facility $x$ to facility $y$ and $D$ is the distance matrix each of whose entries $d_{ij}$ is the distances between location $i$ and location $j$. A valid assignment of facilities to locations is a permutation $\pi: [n] \to [n],$ where $[n]:= \{1, \ldots, n\}$. We can represent $\pi$ as a permutation matrix $\Pi$ each of whose entries $\pi_{xi}$ is 1 if $\pi(x) = i$~(facility $x$ is placed in location $i$) and 0 otherwise. The cost function for the \ac{QAP} can be represented~as 

\[C(\Pi) = \sum_{x,y,i,j=1}^n f_{xy} d_{ij} \pi_{xi} \pi_{yj}.\]

The goal is to find the assignment of facilities to locations satisfying the above constraints that minimizes $C(\Pi).$ For an $n\times n$ \ac{QAP} problem, using brute force one would have to check~$n!$ assignments $\Pi,$ as $|S_n| = n!$. As $n$ grows larger, it becomes exponentially difficult to find the assignment $\pi$ corresponding to the minimum cost $C(\Pi)$ for the $n \times n$ \ac{QAP}.

\subsection{MIS Using Rydberg Arrays}

In \cite{quantumOpt}, the quantum optimization of the MIS problem utilizing Rydberg arrays and QuEra's Aquila hardware showcased promising results for solving the MIS problem with Rydberg arrays. The authors observe a superlinear quantum speedup in finding exact solutions of the MIS problem. In this section, we discuss the formulation using encoding gadgets for the \ac{MIS} and \ac{MWIS} Problems with Rydberg arrays as in \cite{nguyen2023quantum}. We start with the following definitions.

\begin{definition} A \textbf{graph} $G = (V,E)$ is a pair consisting of a vertex set $V(G)$ and an edge set $E(G)$, where an edge is a set of two vertices. 
    \end{definition}

\begin{definition} An \textbf{independent set} of a graph $G$ is a subset of the vertices of $G$ in which no two vertices in the subset are connected by an edge.
\end{definition}

\begin{definition} A \textbf{maximum independent set} (MIS) of a graph $G$ is an independent set containing at least as many vertices as any other independent set.
\end{definition}

\begin{figure}[h!]
  \centering
  \begin{subfigure}[b]{0.45\textwidth}
  \centering
    \includegraphics[width=.8\textwidth]{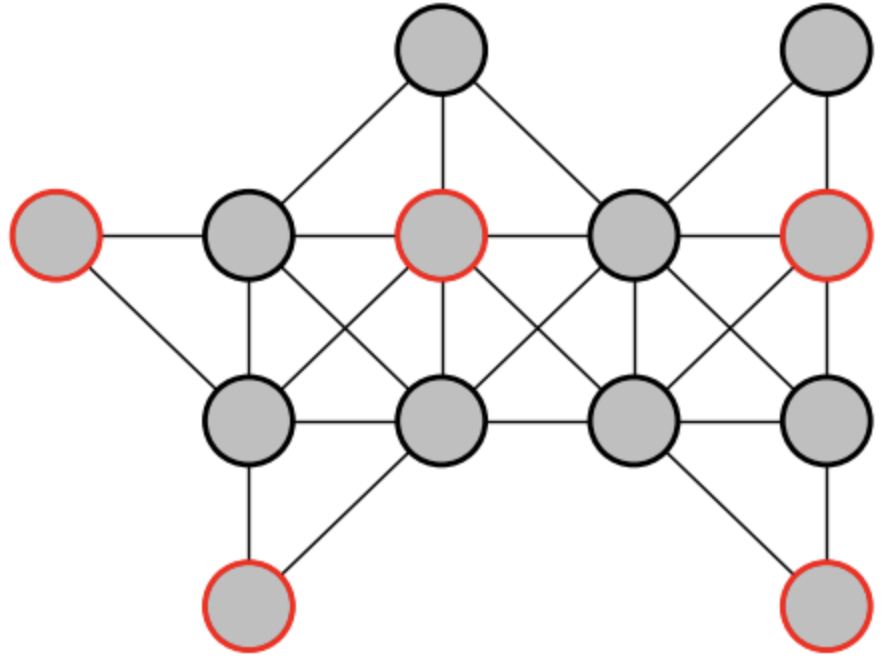}
    \caption{An MIS example}
    \label{fig:MIS}
  \end{subfigure}\hspace{0.03\textwidth}
  \begin{subfigure}[b]{0.45\textwidth}
  \centering
    \includegraphics[width=.8\textwidth]{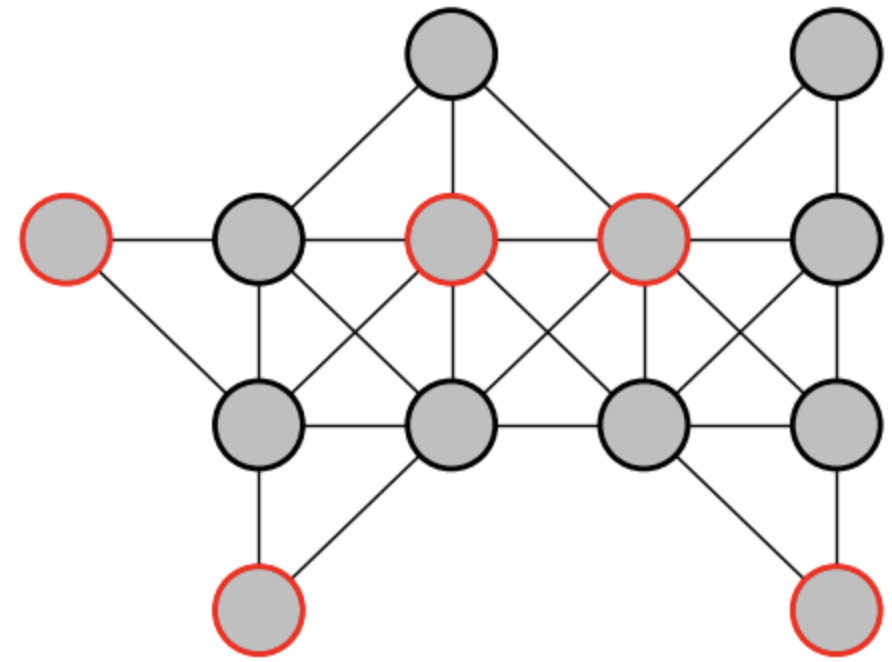}
    \caption{A non MIS example}
    \label{fig:nonMIS}
  \end{subfigure}
    \captionsetup{width=.8\linewidth}
  \caption[width=\textwidth]{Examples of MIS: (a) the set of vertices highlighted in red is the maximum independent set (b) the set of vertices highlighted in red is not an independent set (two vertices share an edge)}
\end{figure}

\noindent We can extend these notions to apply to weighted graphs as well.

\begin{definition}
In a graph $G$ with weighted vertices, the \textbf{weight of a subset} $S$ of the vertices of $G$ is the sum of the weights of the vertices in $S$. A \textbf{maximum weighted independent set}~(MWIS) is an independent set whose weight is at least that of any other independent~set.
\end{definition}

 Given a graph $G$ with $n$ vertices, the computational problem of finding a maximum independent set (or maximum weight independent set) is known as the \ac{MIS} problem (or~\ac{MWIS}, respectively). In general, these problems place no restrictions on the structure of~$G$. We consider what happens when we restrict the inputs to more structured families of graphs, such as the family of unit disk graphs.

\begin{definition}
A \textbf{unit disk graph} (UDG) is a graph in which two vertices share an edge if and only if the vertices lie within a unit distance of each other.
\end{definition}

When the input to a problem is restricted to a UDG, the problem is called UDG-MIS. In \cite{nguyen2023quantum}, the authors introduced three gadgets to encode the \ac{MIS} and \ac{MWIS} problems on Aquila, by converting them to equivalent UDG-MWIS problems. We will now describe their approach by introducing the three gadgets for encoding, along with the crossing lattice. The copy gadget represented in Figure~\ref{fig:gadgets-copy} is a wire carrying a single bit of information. The crossing gadget represented in Figure~\ref{fig:gadgets-xing} allows for two wires to cross without interference. The crossing-with-edge gadget represented in Figure~\ref{fig:gadgets-xwe} allows for two wires to cross, where both cannot have the value $1$. The red borders for the atoms represents an excited atom and the black borders represents an unexcited atom. White corresponds to weight $1\delta,$ grey corresponds to weight $2 \delta$ and black corresponds to weight $4\delta$. For a graph $G = (V,E)$, a crossing lattice is first constructed by using a copy gadget to represent each vertex~$v \in V$ and each line representing a vertex of copy gadgets is drawn with a vertical and a horizontal segment, forming an upper triangular crossing lattice. In this way, each vertical and horizontal line cross exactly once. At each crossing point, the various crossing gadgets will be used to induce interactions between vertices. Using the encoding gadgets introduced in~\cite[Section V]{nguyen2023quantum}, the authors showed that a variety of computational problems can be encoded into UDG-MWIS. The two important steps in encoding these computation problems are as~follows: 

\begin{figure}[h!]
  \centering
\begin{subfigure}[b]{0.25\textwidth}
\includegraphics[width=\textwidth]{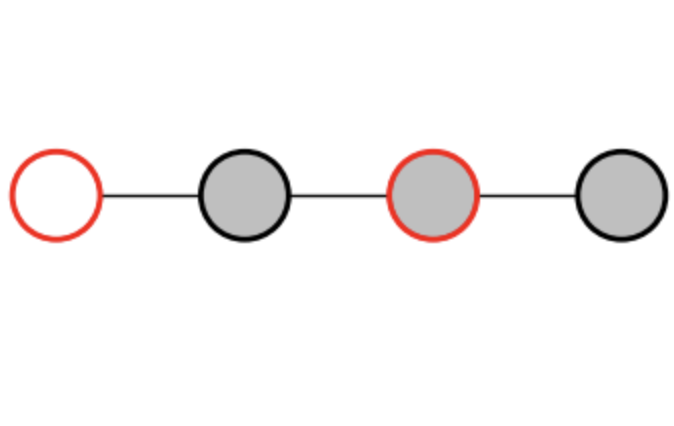}
    \caption{Copy gadget}
    \label{fig:gadgets-copy}
\end{subfigure}\hspace{0.015\textwidth}
\begin{subfigure}[b]{0.25\textwidth}
\includegraphics[width=\textwidth]{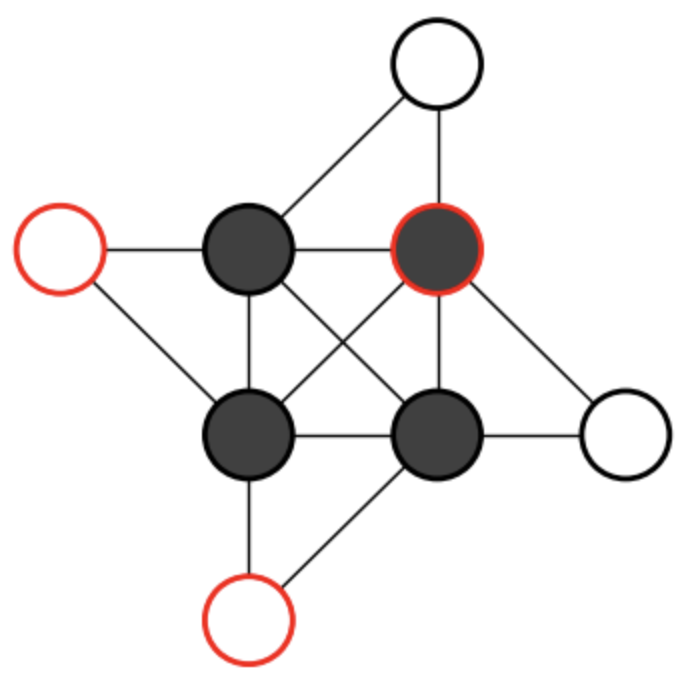}
\caption{Crossing Gadget}
    \label{fig:gadgets-xing}
\end{subfigure}\hspace{0.02\textwidth}
\begin{subfigure}[b]{0.3\textwidth}
\includegraphics[width=.85\textwidth]{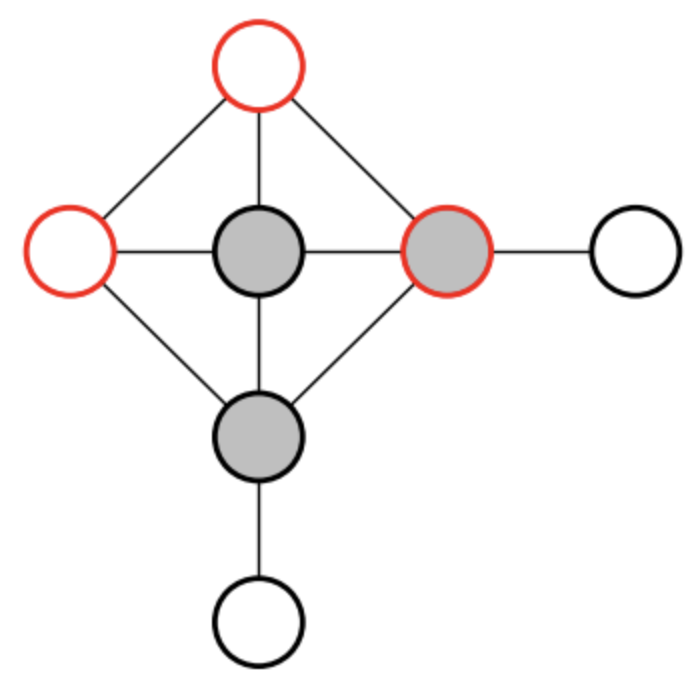}
\caption{Crossing With Edge Gadget}
    \label{fig:gadgets-xwe}
\end{subfigure}
\caption{Gadgets for encoding}
\label{fig:gadgets}
\end{figure}

\begin{enumerate}
    \item The first is to construct a crossing lattice using the copy gadget in Figure \ref{fig:gadgets-copy}.
    \item The second is to apply crossing replacements using the crossing gadget and crossing with edge gadget to encode arbitrary connectivity. 
\end{enumerate}

For a graph $G$ that we want to find the \ac{MWIS}, the crossing gadget is used to encode the crossing between two vertices that do not share an edge. The crossing with edge gadget is used to encode two vertices that share an edge. Using the techniques developed in~\cite[Section~V.B]{nguyen2023quantum}, one can construct the \ac{MWIS} representation of the copy gadget, which consists of a string of $N$ vertices and edges between neighboring vertices. All vertices have a weight~$2 \delta$, except for the two boundary vertices of the line, which have weights $\delta$. By construction, the maximum weight a quantum computer gains from a single defect for the \ac{MWIS} problem is the maximum of flipping a position on a wire $x_i$ with weight $w_i.$ Thus, if $\delta$ is chosen such that $\delta > w_i$, where $w_i$ is the weight of each wire, then an invalid solution does not weigh more than a correct solution. The normalization of vertex weights is needed to ensure that the ground state of the mapped problem used correctly encodes the solution of the original problem. Therefore, the normalization for the \ac{MWIS} problem must obey the constraint \[\delta > \max_i|w_i|.\]

Finding this normalization constraint for other problems is crucial in formulating \ac{MIS} and \ac{MWIS} problems on a quantum computer and will be addressed for QAP when the reduction of QAP to UDG-MWIS is presented in Section~\ref{reduction}.

\section{The ``Carcassonne Computer''}\label{circuit}
The power of Rydberg arrays to solve UDG-MWIS is only useful if we know how to encode problems of interest (\ac{QAP} in particular) as UDG-MWIS problems. We should not hope to find a bespoke encoding for each problem. In this section, we introduce a general framework for building circuits for a broader class of constrained binary optimization problems. This general framework will be used to discuss the \ac{QAP} to UDG-MWIS reduction.

\subsection{Constrained Binary Optimization} \label{sec:tiles-cbop}
The crossing lattice and gadgets introduced in \cite{nguyen2023quantum} provide a systematic approach for reducing various optimization problems to UDG-MWIS. Building on this work, we introduce a visual language as a convenient abstraction both to explain more clearly the reduction for~\ac{QAP} and to facilitate the design of quantum algorithms for other problems. This visual language describes not merely \ac{QAP} but a broader class of \textbf{constrained binary optimization problems}: given a list of $n$ binary variables $x_1, \dots, x_n$, a set $X \subseteq \{0, 1\}^n$ of valid assignments to those variables, and a weight function $w : X \to \R$, find an assignment~$x = (x_1, \dots, x_n) \in X$ that maximizes $w(x)$. We will often find it useful to define $X$ via a list of binary constraints $C_1, \dots, C_m $, where $X = \{x \in \{0, 1\}^n: C_1(x) \wedge \cdots \wedge C_m(x) = 1\}$. We may also (by Lagrange interpolation) represent the weight function as a polynomial of total degree $k \le n$:
$$w(x) = \sum_{b \in \{0, 1\}^n} w_b x_1^{b_1}\cdots x_n^{b_n}.$$

Let $A = (\{x_1, \dots, x_n\}, X, w_A)$ and $B = (\{y_1, \dots, y_m\}, Y, w_B)$ be two constrained binary optimization problems. We say that $A$ \textbf{encodes} $B$ if there exists a  function $f : \{0, 1\}^m \to \{0, 1\}^n$ and constant $w_0$ such that $f(X) = Y$ and $w_B(f(x)) = w_A(x) + w_0$ for all $x \in X$. That is, $f$ relates the variables of $B$ to the variables of $A$ in such a way that solving $A$ amounts to solving $B$: if $x$ is a solution (valid and optimal) to $A$ then $f(x)$ must also be a solution to $B$. In practice, of course, we want $f$ to be easily computable; in the best case it is a simple relabeling of some variables (and possible forgetting of others), of the form~$f(x_1, \dots, x_n) = (x_{k_1}, \dots, x_{k_m})$.

One can formulate an $n \times n$ \ac{QAP} instance quite naturally as a constrained binary optimization problem. The variables $\pi_{xi}$ represent whether facility $x$ is assigned to location~$i$. There are two types of constraints needed: facility constraints $C_{x*}$ to ensure each facility~$x$ is placed in exactly one location, and location constraints $C_{*i}$ ensure each location $i$ is assigned exactly one facility. The weight function can be seen as the negation of the \ac{QAP} cost function, so that the assignment with maximum weight is the one with minimum cost.
\begin{enumerate}
    \item \textbf{Variables:} $\pi_{xi}$ for each $x, i \in [n]$
    \item \textbf{Constraints:} $C_{x*} : \sum_{i=1}^n \pi_{xi} = 1$ and $C_{*i} : \sum_{x=1}^n \pi_{xi} = 1$ for each $x, i \in [n]$
    \item \textbf{Weight Function:}
    $w(\Pi) = -\sum_{x,y,i,j=1}^n f_{xy} d_{ij} \pi_{xi} \pi_{yj}$
\end{enumerate}

While this formulation is accurate and may be the `canonical' formulation of \ac{QAP} as a constrained binary optimization problem, it is not the only one. In other words, there are other constrained binary optimization problems that encode this one. Different formulations, while in some sense equivalent, may introduce more or less overhead when it comes time to construct a UDG-MWIS instance representing the problem. The result is that when a visual language that is flexible enough to represent this broader class of problems is used, we gain not only applicability to problems besides \ac{QAP} but also increased efficiency for the~\ac{QAP}.

\subsection{Circuit Tiles and Weighted Circuits} \label{sec:tiles-tiles}

The basic unit of this language is the \textbf{circuit tile}, a square tile representing a circuit component such as a length of wire or a logic gate, with at most one wire extending to each of its edges. Formally, a circuit tile is described by the set $S \subseteq [4]$ of edges which have a wire and a truth table of possible combinations of wire values $x_i \in \{0, 1\}$. In principle, there are as many different circuit tiles as possible truth tables; in practice, we prefer to use those tiles which have an easily understood meaning. These fall into four basic categories:
\begin{enumerate}
    \item \textbf{Variables} store the value of a named variable; connected wires will carry this value.
    \item \textbf{Wire segments} propagate a value from one edge to another.
    \item \textbf{Wire intersections} allow two wires to cross without interference.
    \item \textbf{Logic gates} set the output wire(s) as a logical combination of the input wire(s).
\end{enumerate}
Figure \ref{fig:tiles} gives an example of each type of tile with its associated truth table (where edges are numbered counterclockwise from the right edge).

\begin{figure}[h!]
    \centering
    \begin{subfigure}[b]{0.2\linewidth}
        \centering
        \begin{tabular}{|c|}
            \hline
            $x_2$ \\
            \hline
            0 \\
            1 \\
            \hline
        \end{tabular}
        \par\medskip
        \includegraphics[width=0.7\linewidth]{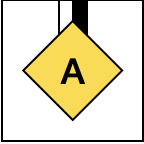}
        \caption{Variable}
    \end{subfigure}
    \begin{subfigure}[b]{0.2\linewidth}
        \centering
        \begin{tabular}{|c|c|}
            \hline
            $x_2$ & $x_4$ \\
            \hline
            0 & 0 \\
            1 & 1 \\
            \hline
        \end{tabular}
        \par\medskip
        \includegraphics[width=0.7\linewidth]{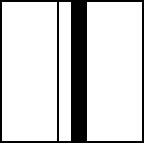}
        \caption{Wire Segment}
    \end{subfigure}
    \begin{subfigure}[b]{0.2\linewidth}
        \centering
        \begin{tabular}{|c|c|c|c|}
            \hline
            $x_1$ & $x_2$ & $x_3$ & $x_4$ \\
            \hline
            0 & 0 & 0 & 0 \\
            0 & 1 & 0 & 1 \\
            1 & 0 & 1 & 0 \\
            1 & 1 & 1 & 1 \\
            \hline
        \end{tabular}
        \par\medskip
        \includegraphics[width=0.7\linewidth]{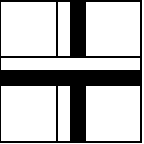}
        \caption{Intersection}
    \end{subfigure}
    \begin{subfigure}[b]{0.2\linewidth}
        \centering
        \begin{tabular}{|c|c|c|}
            \hline
            $x_1$ & $x_2$ & $x_3$ \\
            \hline
            0 & 0 & 0 \\
            0 & 1 & 1 \\
            1 & 0 & 1 \\
            1 & 1 & 1 \\
            \hline
        \end{tabular}
        \par\medskip
        \includegraphics[width=0.7\linewidth]{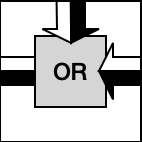}
        \caption{OR Gate}
    \end{subfigure}
    \caption{Truth tables and visual depictions of common circuit tiles}
    \label{fig:tiles}
\end{figure}

Given a tile, we can apply \textbf{decorations} to change its function. Each decoration affects one assignment of wire values (\textit{i.e.} one row of the tile's truth table), indicated by the placement of the decoration on the tile. Each wire is drawn with its top/left side light to indicate a value of 0 and its bottom/right side dark to indicate a value of 1; a decoration placed on one side of the wire affects the assignments in which that wire carries the corresponding value. On a tile with both a vertical and a horizontal wire (for a logic gate, this means the input wires), a decoration on side $i$ of the horizontal wire and side $j$ of the vertical wire is called an $(i,j)$-decoration. There are two types of decorations with different effects, both shown in Figure \ref{fig:decor}.
\begin{enumerate}
    \item \textbf{Restrictions} (indicated by a red ``X'') prohibit a certain assignment of wire values, essentially deleting a row from the tile's truth table.
    \item \textbf{Biases} (indicated by a blue circle) incentivize or dis-incentivize a certain assignment of the wire values by applying a set weight when that assignment is met.
\end{enumerate}

\begin{figure}[ht]
    \centering
    \begin{subfigure}[b]{0.3\linewidth}
        \centering
        \begin{tabular}{|c|c|c|c|}
            \hline
            $x_1$ & $x_2$ & $x_3$ & $x_4$ \\
            \hline
            0 & 0 & 0 & 0 \\
            0 & 1 & 0 & 1 \\
            1 & 0 & 1 & 0 \\
            1 & 1 & 1 & 1 \\
            \hline
        \end{tabular}
        \par\medskip
        \includegraphics[width=0.5\linewidth]{intersection.png}
        \caption{Original Tile}
        \label{fig:decor-basic}
    \end{subfigure}
    \begin{subfigure}[b]{0.3\linewidth}
        \centering
        \begin{tabular}{|c|c|c|c|}
            \hline
            $x_1$ & $x_2$ & $x_3$ & $x_4$ \\
            \hline
            0 & 0 & 0 & 0 \\
            0 & 1 & 0 & 1 \\
            1 & 0 & 1 & 0 \\
            \hline
        \end{tabular}
        \par\medskip
        \includegraphics[width=0.5\linewidth]{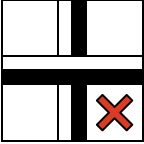}
        \caption{With (1,1)-Restriction}
        \label{fig:decor-restrict}
    \end{subfigure}
    \begin{subfigure}[b]{0.3\linewidth}
        \centering
        \begin{tabular}{|c|c|c|c|}
            \hline
            $x_1$ & $x_2$ & $x_3$ & $x_4$ \\
            \hline
            0 & 0 & 0 & 0 \\
            \rowcolor{softblue} 0 & 1 & 0 & 1 \\
            1 & 0 & 1 & 0 \\
            1 & 1 & 1 & 1 \\
            \hline
        \end{tabular}
        \par\medskip
        \includegraphics[width=0.5\linewidth]{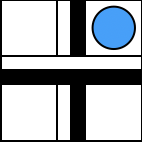}
        \caption{With (0,1)-Bias}
        \label{fig:decor-bias}
    \end{subfigure}
    \caption{Adding decorations to an intersection tile}
    \label{fig:decor}
\end{figure}

Observe that, while the intersection tile in Figure \ref{fig:decor-basic} is the circuit tile analogue of the crossing gadget, the (1,1)-restricted tile in Figure \ref{fig:decor-restrict} is the analogue of the crossing-with-edge gadget. This analogy will become even more apparent when we discuss the process of compiling circuit tiles to unit disk graphs in Section \ref{sec:tiles-compile}. Once we have a library of circuit tiles we can assemble them into a \textbf{weighted circuit}, connecting wires to wires and empty edges to empty edges as follows. This is a reminiscent of the board game \textit{Carcassonne}. 
A \textbf{valid assignment} to the circuit is an assignment of wire values which is consistent both across edges (connected wires must have the same value) and within tiles (all truth tables must be obeyed). The \textbf{weight} of an assignment is the sum of all weights applied on that assignment.

In this way, a weighted circuit $C$ naturally defines a constrained binary optimization problem $A_C$ whose variables are the values of each wire of each tile, with constraints and weight function are described above. We say that $C$ \textit{encodes} a constrained binary optimization problem $B$ if $A_C$ encodes $B$. Example \ref{ex:circuit-ex} shows how a weighted circuit may encode a constrained binary optimization problems.

\begin{example} \label{ex:circuit-ex}
Consider the following constrained binary optimization problem, along with a weighted circuit (see Figure \ref{fig:circuit-ex}) which encodes it:
\begin{enumerate}
    \item \textbf{Variables:} $x_1, x_2, x_3$
    \item \textbf{Constraints:} $C_{ij}: x_i \vee x_j = 1$ for each $1 \le i < j \le 3$
    \item \textbf{Weight Function:} $w(x) = w_1x_1 + w_2x_2 + w_3x_3 + w_4x_1x_2x_3$
\end{enumerate}

\begin{figure}[ht]
    \centering
    \includegraphics[width=0.37\linewidth]{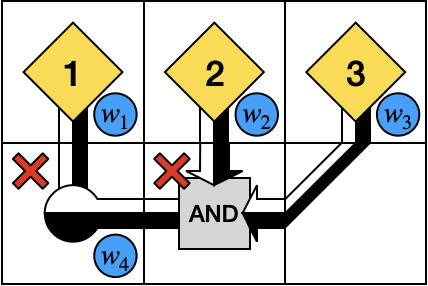}
    \captionsetup{width=.8\linewidth}
    \caption{A simple weighted circuit. The bottom-left tile represents two wires which intersect at their endpoints, allowing for interaction.}
    \label{fig:circuit-ex}
\end{figure}
\end{example}

In the weighted circuit, each variable $x_i$ is encoded by wire $i$. The (0,0)-restriction on the AND tile enforces $C_{23}$, while the (0,0)-restriction on the crossing tile enforces $C_{12} \wedge C_{13}$: the vertical wire carries value $x_1$, the horizontal wire carries value $x_2 \wedge x_3$, and the (0,0)-restriction guarantees that $x_1 \vee (x_2 \wedge x_3) = (x_1 \vee x_2) \wedge (x_1 \vee x_3) = 1$. Therefore, the valid assignments to the circuit are precisely those which satisfy all three constraints $C_{ij}$. Moreover, the biases in the circuit encode the weight function $w$: For $1 \ge i \ge 3$ the weight~$w_i$ is applied when~$x_i = 1$, and the $(1,1)$-bias on the crossing tile applies $w_4$ when we have that~$x_1 \wedge (x_2 \wedge x_3) = x_1x_2x_3 = 1$.

\subsection{Compilation to UDG-MWIS}
\label{sec:tiles-compile}

Circuit tiles exist to provide a shortcut from constrained binary optimization problems to UDG-MWIS. After encoding a constrained binary optimization problem $P$ as a weighted circuit $C$, we want to ``compile'' $C$ to a graph $G$ whose maximum-weight independent sets correspond to optimal assignments for $C$ (and thus for $P$ as well). For a large and complex circuit, it may be quite difficult to construct such a graph by hand; for a small and simple circuit, however, it should be much easier. We will start with the simplest circuits of all---individual tiles---and show how to build graph compilations of more complex circuits using the compilations of their component tiles.

We compile each circuit tile $T$ to a weighted subgraph $G_T$ of the $4 \times 4$ king's graph, as shown in Figure \ref{fig:udg-tiles}. Similarly, we compile a circuit $C$ with an $m \times n$ grid of tiles compiles to a weighted subgraph $G_C$ of the $4m \times 4n$ king's graph, simply by compiling each individual tile and stitching the resulting graphs together. Intuitively, we want the maximum-weight independent sets of each graph $G_T$ to correspond to valid assignments for the corresponding tile $T$, in hopes that if this property holds for each individual tile then it will hold for the entire circuit. It turns out this is not quite the correctness property we need, due to some technicalities which will appear later. In the rest of this section we  will define exactly what makes $G_T$ a ``correct'' compilation of $T$, then show that the correctness of each tile compilation guarantees the correctness of the circuit-to-UDG-MWIS reduction as a whole.

\begin{figure}[ht]
    \centering
    \begin{subfigure}[t]{0.3\linewidth}
        \centering
        \includegraphics[width=0.5\linewidth]{intersection.png}
        \par\medskip
        \includegraphics[width=0.5\linewidth]{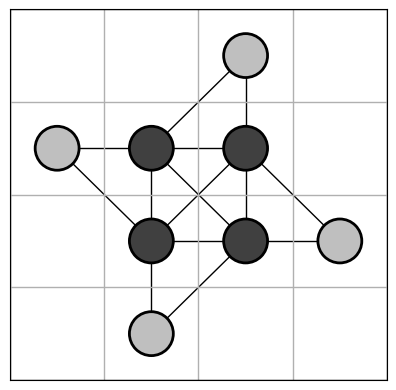}
        \caption{Intersection}
        \label{fig:udg-00and}
    \end{subfigure}
    \begin{subfigure}[t]{0.3\linewidth}
        \centering
        \includegraphics[width=0.5\linewidth]{11.png}
        \par\medskip
        \includegraphics[width=0.5\linewidth]{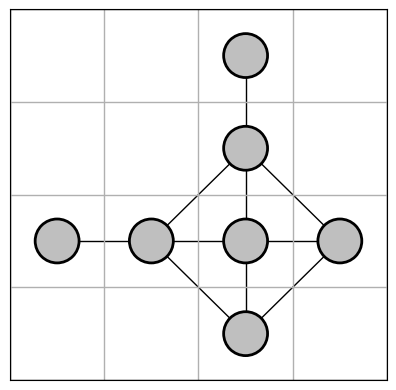}
        \caption{\centering $(1,1)$-Restricted Intersection}
        \label{fig:udg-11intersect}
    \end{subfigure}
    \begin{subfigure}[t]{0.3\linewidth}
        \centering
        \includegraphics[width=0.5\linewidth]{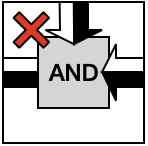}
        \par\medskip
        \includegraphics[width=0.5\linewidth]{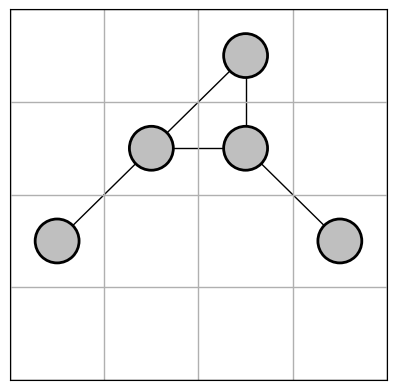}
        \caption{\centering $(0,0)$-Restricted AND Gate}
        \label{fig:udg-intersect}
    \end{subfigure}
    \caption{UDG-MWIS compilations of selected tiles.}
    \label{fig:udg-tiles}
\end{figure}

\hfill\par\noindent\textbf{Connecting and internal vertices.} First, we require the graph $G_T$ of each tile $T$ to have exactly one vertex $v_i$ on each edge $i$ for which $T$ a wire (\textit{e.g.} if $T$ has a wire on edge $1$ then $G_T$ must include exactly one of the vertices in the top row of the $4 \times 4$ king's graph). We call $v_i$ the \textbf{connecting vertex} on edge $i$, and $V_\mathrm{con}(G_T)$ the set of connecting vertices of~$G_T$. In this section, we will assume that connecting vertices are not placed in the corners of the~$4 \times 4$ kings graph (in the full circuit graph this ensures that the connecting vertices of two adjacent tiles are themselves adjacent and avoids interference between diagonally adjacent tiles). A vertex which is not a connecting vertex we call an \textbf{internal vertex}; the set of such vertices we denote $V_\mathrm{int}(G_T)$.

These notions can also be extended to circuits. Consider a circuit $C$ with graph $G_C$. Visually, the connecting vertices of $G_C$ are the vertices that correspond to the ``dangling'' wires of $C$. Formally, the connecting vertices of $G_C$ are all the connecting vertices of its component tiles except those that are adjacent to a connecting vertex of another tile. Let~$V_\mathrm{con}(G_C)$ be the set of its connecting vertices. If $V_\mathrm{con}(G_C) = \emptyset$, we say that the circuit $C$ is \textbf{closed}. For correctness of compilation we assume the circuit to be compiled is closed.

\hfill\par\noindent\textbf{Analysis of independent sets.} Given a tile $T$ with graph $G_T$, an independent set~$S \subseteq V(G_T)$ defines the following assignment $x(S)$ of wire values to $T$:
\begin{itemize}
    \item For edges $1$ (right) and $4$ (bottom), wire $i$ has value $x_i = 1$ if $v_i \in S$ and $x_i = 0$ if $v_i \notin S$.
    \item For edges $2$ (top) and $3$ (left), it is reversed: $x_i = 0$ if $v_i \in S$ and $x_i = 1$ if $v_i \notin S$.
\end{itemize}
Observe that the assignment is reversed for opposite edges. In the completed circuit, each edge between tiles (crossed by a wire) will have a connecting vertex on either side and each maximum-weight independent set will contain exactly one vertex from each pair of connecting vertices. In this setting, the reversal gives consistency of wire values across edges. Given a circuit $C$ and an independent set $S \subseteq V(G_C)$, the assignment $x(S)$ is defined by giving each tile $T$ in $C$ the assignment $x(S \cap V(G_T))$.

Given a circuit $C$ with graph $G_C$, we define the \textbf{circuit weight} $w_\mathrm{circ}$ of an independent set $S \subseteq V(G_C)$ as
$$w_\mathrm{circ}(S, G_C) = \sum_{v \in S} w(v) + \delta \cdot |V_\mathrm{con}(G_C) \setminus S|$$
for some constant $\delta > 0$. Strictly speaking, $G_C = G_C(\delta)$ is a family of graphs parametrized by $\delta$. For simplicity, we will write it as a single graph whose vertex weights are defined in terms of $\delta$. This is similar to the weight of the independent set but it also counts the number of connecting vertices excluded from the independent set. These vertices are important when~$C$ is a subcircuit of a larger circuit, because they allow $S$ to extend to include connecting vertices of adjacent tiles. It follows that, for a closed circuit~$C$, the circuit weight~$w_\mathrm{circ}(S, G_C)$ of an independent set $S \subseteq V(G_C)$ is identical to its weight~$w(S) = \sum_{v \in S} w(v)$.

Let $C$ be a weighted circuit with graph $G_C$, and let $C'$ be a subcircuit of $C$. Recall that we can divide $G_C$ into $4 \times 4$ boxes corresponding to the tiles of $C$. The \textbf{induced subcircuit graph} of $C'$ is the subgraph $G'$ of $G_C$ containing only those vertices in boxes that correspond to tiles of $C'$. Let $C_1$ and $C_2$ be two subcircuits of $C$, with induced subcircuit graphs $G_1$ and $G_2$. We say that $C$ \textbf{splits into} $C_1$ and $C_2$ if every tile in $C$ is contained in either $C_1$ or $C_2$, but not both. In this case we define the \textbf{boundary} of $G_1$ as $$B(G_1) = \{v_1 \in V(G_1) : \exists v_2 \in V(G_2), \{v_1, v_2\} \in E(G_C)\},$$ and likewise the boundary $B(G_2)$ of $G_2$. We also define the \textbf{shared boundary} of $G_1$ and $G_2$ as $B(G_1, G_2) = B(G_1) \cup B(G_2)$. We will use the shorthand notation $B_1$ for $B(G_1)$, $B_2$ for $B(G_2)$, and $B$ for $B(G_1, G_2)$. We have a useful lemma allowing us to understand the circuit weight of $S \subseteq V(G_C)$ in $C$ by studying its restriction to subcircuits of $C$.

\begin{lemma} \label{lem:weight}
    Let $C$ be a weighted circuit with graph $G_C$, which splits into subcircuits $C_1$ and $C_2$ with induced subcircuit graphs $G_1$ and $G_2$. Let $S$ be an independent set of $G_C$, with $S_1 = S \cap V(G_1)$ and $S_2 = S \cap V(G_2)$. We have that
    $$w_\mathrm{circ}(S, G_C) = w_\mathrm{circ}(S_1, G_1) + w_\mathrm{circ}(S_2, G_2) - \delta \cdot |B \setminus S|.$$
\end{lemma}

\begin{proof}
    Observe that $V_\mathrm{con}(G_C) \sqcup B = V_\mathrm{con}(G_1) \sqcup V_\mathrm{con}(G_1)$, and thus
    $$(V_\mathrm{con}(G_C) \setminus S) \sqcup (B \setminus S) = (V_\mathrm{con}(G_1) \setminus S_1) \sqcup (V_\mathrm{con}(G_2) \setminus S_2).$$
    It follows that
    \begin{align*}
        w_\mathrm{circ}(S, G_C)
        &= w(S) + \delta \cdot |V_\mathrm{con}(G_C) \setminus S| \\
        &= w(S_1) + w(S_2)
            + \delta \cdot |V_\mathrm{con}(G_1) \setminus S_1|
            + \delta \cdot |V_\mathrm{con}(G_2) \setminus S_2|
            - \delta \cdot |B(G_1, G_2) \setminus S| \\
        &= w_\mathrm{circ}(S_1, G_1) + w_\mathrm{circ}(S_1, G_1)
            - \delta \cdot |B(G_1, G_2) \setminus S|
    \end{align*}
    which concludes the proof.
\end{proof}

We now define correctness of a compilation. 

\begin{definition} Let $C$ be a weighted circuit (which could be a single tile $T$), with valid assignments $X$ and weight function $w : X \to \R$. We call $G_C$ a \textbf{correct compilation} of circuit $C$ if it has the proper connecting vertices and there exists some constants $k \in Z$ and~$\tilde{w} \ge \max_{x \in X} |w(x)|$, such that the following properties are satisfied for all $\delta > 0$:
\begin{enumerate}
    \item For all $x \in X$, there is an independent set $S \subseteq V(G_C)$ with $x(S) = x$ and $w_\mathrm{circ}(S) = k\delta + w(x)$.
    \item For any independent set $S \subseteq V(G_C)$ with $x = x(S)$, if $x \in X$, then $w_\mathrm{circ}(S) \le k\delta + w(x)$ and if $x \notin X$ then $w_\mathrm{circ}(S) \le (k - 1)\delta + \tilde w$.
\end{enumerate}
\end{definition}

This definition of correctness supports a modular approach to circuit design: correct compilations for the subcircuits of a given circuit give us a correct compilation for the entire circuit.

\begin{lemma} \label{lem:correctness}
    Let $C$ be a weighted circuit with graph $G_C$, which splits into subcircuits $C_1$ and $C_2$ with induced subcircuit graphs $G_1$ and $G_2$. If $G_1$ and $G_2$ are both correct compilations of their corresponding circuits, then $G_C$ is a correct compilation of $C$.
\end{lemma}

\begin{proof}
    By construction, $G_C = (E, V, w)$ has the proper connecting vertices. We know that~$G_1$ and $G_2$ satisfy the two correctness properties for some constants $\tilde{w}_1$, $\tilde{w}_2$, $k_1$, and $k_2$. Hence, we define $k = k_1 + k_2 - \frac{1}{2}|B|$ and $\tilde{w} = \tilde{w}_1 + \tilde{w}_2$. Note that $k$ is an integer since~$\frac{1}{2}|B| = |B_1| = |B_2| \in \Z$, and that $\tilde{w} \ge \max_{x_1 \in X_1} |w_1(x_1)| + \max_{x_2 \in X_2} |w_2(x_2)| \ge \max_{x \in X} |w(x)|$. Now we will show, for any $\delta > 0$, that $G_C$, $k$, and $\tilde w$ satisfy the two properties of a correct compilation.

    \hfill\par\noindent\textbf{Property 1:} Let $x$ be a valid assignment of wires to $C$. We can write $x = (x_1, x_2)$, where~$x_1 \in X_1$ is a valid assignment of wires to $C_1$ and $x_2 \in X_2$ is a valid assignment of wires to $C_2$, Moreover, $x_1$ and $x_2$ must be consistent with each other: if a wire in $C_1$ connects to a wire in $C_2$, the two wires must be assigned the same value. By the correctness of $G_1$ and $G_2$, there are independent sets $S_1 \subseteq V(G_1)$ and $S_2 \subseteq V(G_2)$ such that $x(S_i) = x_i$ and $w_\mathrm{circ}(S_i) = k_i\delta + w(x_i)$ for each $i \in \{1, 2\}$. Define $S = S_1 \cup S_2$; we will show that $S$ is an independent set of $G_C$, that $x(S) = x$, and that $w_\mathrm{circ}(S, G_C) = k\delta + w(x)$.
    
    First, we show that $S$ is an independent set of $G_C$. Since $S_1$ and $S_2$ are both independent sets, it is sufficient to show that $\{v_1, v_2\} \notin E$, for all $(v_1, v_2) \in S_1 \times S_2$. We let $(v_1, v_2) \in V(G_1) \times V(G_2)$ be such that $\{v_1, v_2\} \in E(G_C)$. It is easy to see that $v_1$ and $v_2$ must be connecting vertices of $G_1$ and $G_2$, respectively. These vertices represent a pair of wires in $C$ which are connected across a tile edge and must therefore be assigned the same value. From the definition of $x(S_1)$ and $x(S_2)$, either $v_1 \in S_1$ and $v_2 \notin S_2$ or $v_1 \notin S_1$ and $v_2 \in S_2$. In either case, $(v_1, v_2) \notin S_1 \times S_2$; this shows that $S$ is indeed an independent set. Next, we show that $x(S) = x$. For every $v \in V(G_1)$, since $S \cap V(G_1) = S_1$ we know that $v \in S$ if and only if $v \in S_2$. Therefore, if there is a wire in $C_1$ that corresponds to~$V$, that wire must be assigned the same value in $x(S)$ as in $x_1 = x(S_1)$. By the same argument, every wire in $C_2$ must be assigned the same value in $x(S)$ as in $x_2 = x(S_2)$. Therefore,~$x(S) = (x_1, x_2) = x$. Finally, we show that $w_\mathrm{circ}(S, G_C) = k\delta + w(x)$. As shown above, for each $\{v_1, v_2\} \in E$, if $(v_1, v_2) \in B_1 \times B_2$ then either $v_1 \in S_1$ and $v_2 \notin S_2$ or $v_1 \notin S_1$ and $v_2 \in S_2$. We have that~$|B \setminus S| = |B \cap S| = \frac{1}{2} |B|$. It follows, by Lemma \ref{lem:weight}, that
    \begin{align*}
        w_\mathrm{circ}(S, G_C)
        &= w_\mathrm{circ}(S_1, G_1) + w_\mathrm{circ}(S_1, G_1)
            - \delta \cdot |B \setminus S| \\
        &= k_1\delta + w_1(x_1) + k_2\delta + w_2(x_2) 
            - \delta \cdot \frac{1}{2}|B| \\
        &= (k_1 + k_2 - \frac{1}{2}|B|)\delta + w_1(x_1) + w_2(x_2) \\
        &= k\delta + w(x).
    \end{align*}

    \hfill\par\noindent\textbf{Property 2:} Let $S$ be an independent set of $G_C$. Define $S_1 = S \cap V(G_1)$ and $S_2 = S \cap V(G_2)$. If $x_1 = x(S_1) \in X_1$ then $w_\mathrm{circ}(S_1, G_1) \le k_1\delta + w_1(x_1) \le k_1\delta + \tilde{w}_1$; on the other hand, if~$x_1 \notin X_1$ then $w_\mathrm{circ}(S_1, G_1) \le (k_1 - 1)\delta + \tilde{w}_1$. Either way, $w_\mathrm{circ}(S_1, G_1) \le k_1\delta + \tilde{w}_1$, and similarly $w_\mathrm{circ}(S_2, G_2) \le k_2\delta + \tilde{w}_2$. Moreover, by assumption each vertex $v_1 \in B_1$ is adjacent to a unique vertex $v_2 \in B_2$, and vice versa. Since $B \cap S$ can include at most one vertex from each pair $(v_1, v_2)$, $|B \cap S| \le \frac{1}{2}|B|$ and thus $|B \setminus S| \ge \frac{1}{2}|B|$. If $x = x(S) \in X$ is a valid assignment for $C$, then $x_1 = x(S_1)$ and $x_2 = x(S_2)$ must both be valid assignments for $C_1$ and $C_2$, respectively. Therefore, by Lemma \ref{lem:weight}:
    \begin{align*}
        w_\mathrm{circ}(S, G_C)
        &= w_\mathrm{circ}(S_1, G_1) + w_\mathrm{circ}(S_1, G_1)
            - \delta \cdot |B \setminus S| \\
        &\le k_1\delta + w_1 + k_2\delta + w_2(x_2)
            - \delta \cdot \frac{1}{2}|B| \\
        &= k\delta + w(x).
    \end{align*}
    If $x(S) \notin X$ is an invalid assignment for $C$, then at least one of the following must be true:
    \begin{enumerate}
        \item $x_1 = x(S_1) \notin X_1$; that is, $x_1$ is an invalid assignment for $C_1$.
        \item $x_2 = x(S_2) \notin X_2$; that is, $x_2$ is an invalid assignment for $C_2$.
        \item $x_1$ and $x_2$ are not consistent with each other.
    \end{enumerate}
    We will consider these three cases individually.
    
    \noindent\textbf{Case 1:} Let $x_1 \notin X_1$. By correctness of $G_1$, $w_\mathrm{circ}(S_1, G_1) \le (k_1 - 1)\delta + \tilde{w}_1$. As shown above, $w_\mathrm{circ}(S_2, G_2) \le k_2\delta + \tilde{w}_2$ and $|B \setminus S| \ge \frac{1}{2}|B|$. Therefore, by Lemma \ref{lem:weight}:
    \begin{align*}
        w_\mathrm{circ}(S, G_C)
        &= w_\mathrm{circ}(S_1, G_1) + w_\mathrm{circ}(S_1, G_1)
            - \delta \cdot |B \setminus S| \\
        &\le (k_1 - 1)\delta + \tilde{w}_1 + k_2\delta + \tilde{w}_2
            - \delta \cdot \frac{1}{2}|B| \\
        &= (k - 1)\delta + \tilde{w}.
    \end{align*}

    \noindent\textbf{Case 2:} Let $x_2 \notin X_2$. By the same argument as in case 1, $w_\mathrm{circ}(S, G_C) \le (k - 1)\delta + \tilde{w}$.

   \noindent\textbf{Case 3:} Suppose that there is some tile edge between $C_1$ and $C_2$ for which the corresponding wire values in $x_1$ and $x_2$ do not agree. Let $v_1 \in B_1$ and $v_2 \in B_2$ be the vertices corresponding to these edges. By the definition of $x(S)$, either $v_1, v_2 \in S$ or $v_1, v_2 \notin S$. Since the two vertices are adjacent, the independent set $S$ cannot contain them both; it must therefore contain neither. As argued above, $B \cap S$ can contain at most one vertex from each other pair of adjacent vertices $v_1' \in B_1, v_2' \in B_2$. Since there are $\frac{1}{2}|B| - 1$ such pairs, $|B \cap S| \le \frac{1}{2}|B| - 1$ and thus $|B \setminus S| \ge \frac{1}{2}|B| + 1$. Therefore, once again by Lemma \ref{lem:weight}:
    \begin{align*}
        w_\mathrm{circ}(S, G_C)
        &= w_\mathrm{circ}(S_1, G_1) + w_\mathrm{circ}(S_1, G_1)
            - \delta \cdot |B \setminus S| \\
        &\le k_1\delta + \tilde{w}_1 + k_2\delta + \tilde{w}_2
            - \delta \cdot \left( \frac{1}{2}|B| + 1 \right) \\
        &= (k - 1)\delta + \tilde{w}.\qedhere
    \end{align*}
\end{proof}

By repeated application of Lemma \ref{lem:correctness}, we can continue to split $C$ until each subcircuit is a single tile. This leads us to Theorem \ref{thm:correctness}, with the outcome that, given a library of correct compilations for individual circuit tiles, we can encode any circuit constructed from those tiles as a UDG-MWIS instance.

\begin{theorem}\label{thm:correctness}
    Let $C$ be a closed circuit with graph $G_C$. For each tile $T_1, \dots T_n$ of $C$, we define $G_i$ as the induced subcircuit graph of $T_i$. If each $G_i$ is a correct compilation for~$T_i$ with parameters $k_i$ and $\tilde{w}_i$, then $G_C$ is a correct compilation of $C$ with parameters~$k~=~\sum_{i=1}^n \left( k_i - \frac{1}{2}|V_\mathrm{con}(G_i)| \right)$ and $\tilde{w} = \sum_{i=1}^n \tilde{w}_i$. Moreover, when $\delta > 2 \tilde{w}$, the maximum-weight independent set(s) of $G_C$ correspond to the maximum-weight valid assignment(s) of $C$ (that is, the map $x(S)$ defines a surjection from the set of maximum-weight independent sets of~$G_C$ to the set of maximum-weight valid assignments of $C$).
\end{theorem}

\begin{proof}
    The correctness of $G_C$ and the formula for $\tilde{w}$ follow immediately from Lemma \ref{lem:correctness}, applied~$n - 1$ times as we build $C$ one tile at a time. Now we show the formula for $k$. Denote by $C_1^t$ and~$C_2^t$ the subcircuits appearing at the $t$-th step, with induced subcircuit graphs~$G_1^t$ and $G_2^t$ and shared boundary $B^t = B(G_1^t, G_2^t)$. Since $C$ is closed, every connecting vertex~$v_{ij}$ of every tile $T_i$ must appear in exactly one of $B^1, \dots, B^{n-1}$, which means we have that $\bigsqcup_{i=1}^n V_\mathrm{con}(G_i) = \bigsqcup_{t=1}^m B^t$. Therefore,
    $$k = \sum_{i=1}^n k_i - \frac{1}{2}\sum_{t=1}^m \left| B^t \right| = \sum_{i=1}^n \left( k_i - \frac{1}{2}|V_\mathrm{con}(G_i)| \right).$$
    
    Another consequence of $C$ begin closed is that $V_\mathrm{con}(G_C) = \emptyset$, thus, for any independent set $S \subseteq V(G_C)$, its circuit weight $w_\mathrm{circ}(S, G_C) = w(S) + \delta \cdot |V_\mathrm{con}(G_C) \setminus S|$ is identical to its actual weight $w(S)$. Let $x^*$ be a maximum-weight valid assignment to $C$. By correctness of $G_C$, there exists an independent set $S^* \subseteq V(G_C)$ with $x(S^*) = x^*$ and~$w(S^*) = w_\mathrm{circ}(S^*, G_C) = k\delta + w(x^*)$. Moreover, we claim that $w(S) \le k\delta + w(x^*)$ for any independent set $S \subseteq V(G_C)$, meaning $S^*$ is a maximum-weight independent set of~$G_C$. To show this claim we consider two cases:
    \begin{itemize}
        \item If $x \notin X$ then $w(S) =  w_\mathrm{circ}(S, G_C) \le (k - 1)\delta + \tilde{w} < k\delta - \tilde{w} \le k\delta + w(x^*)$.
        \item If $x \in X$ then $w(S) =  w_\mathrm{circ}(S, G_C) \le k\delta + w(x) \le k\delta + w(x^*)$.
    \end{itemize}
    Therefore, every maximum-weight valid assignment~$x^*$ corresponds to a (possibly non-unique) maximum-weight independent set $S^*$. On the other hand, we claim that every maximum-weight independent set $S^*$ corresponds to a (unique) maximum-weight valid assignment~$x^*$. If $S^*$ is a maximum-weight independent set of $G_C$, with corresponding assignment $x^* = x(S^*) \in X$, then for every $x \in X$ there is an independent set $S \subseteq V(G_C)$ such that $x(S) = x$ and $k\delta + w(x) = w(S) \le w(S^*) \le k\delta + w(x^*)$. So $w(x) \le w(x^*)$, meaning $x^*$ is a maximum-weight valid assignment.
\end{proof}
It is worth noting that the conditions in Theorem \ref{thm:correctness} are sufficient, but not necessary, to ensure the correctness of $G_C$ and the correspondence of its maximum-weight independent sets with the optimal assignments of $C$. In particular, $G_C$ may be a correct compilation for a  smaller $\tilde{w}$ than the one computed in the theorem, allowing for a smaller $\delta$ compared to the weights $w(x)$ and thus for a larger proportional difference between weights of independent sets corresponding to valid wire assignments. Example \ref{ex:correctness-ex} shows how a circuit may have a better compilation than Theorem \ref{thm:correctness} suggests.

\begin{example} \label{ex:correctness-ex}
    Consider the weighted circuit $C$ composed of two unbiased wire tiles $T_1$ and~$T_2$, shown with a compilation $G_C$ in Figure \ref{fig:correctness-ex}. Let $G_1$ and $G_2$ be the induced subcircuit graphs of $T_1$ and $T_2$.

    \begin{figure}[ht]
        \centering
        \begin{subfigure}[t]{0.3\linewidth}
            \centering
            \includegraphics[width=0.9\linewidth]{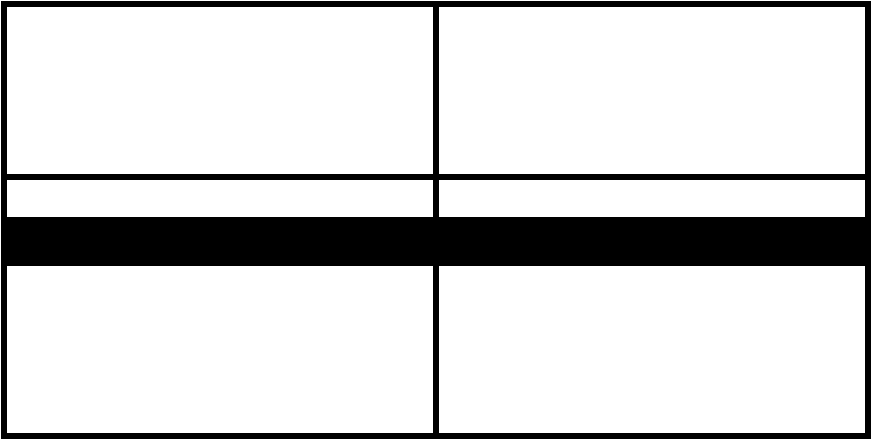}
            \caption{Weighted circuit $C$}
            \label{fig:correctness-ex-circ}
        \end{subfigure}
        \begin{subfigure}[t]{0.3\linewidth}
            \centering
            \includegraphics[width=0.9\linewidth]{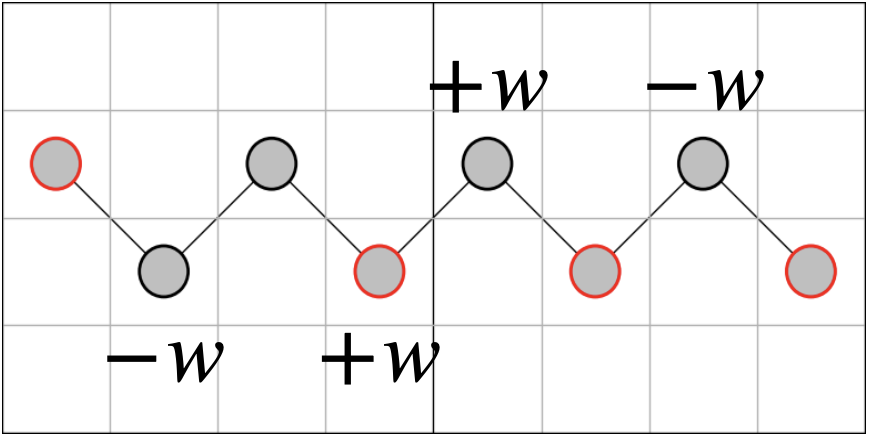}
            \caption{UDG compilation $G_C$}
            \label{fig:correctness-ex-udg}
        \end{subfigure}
        \caption{Correct compilation with $\tilde{w}$ smaller than given by Theorem \ref{thm:correctness}}
        \label{fig:correctness-ex}
    \end{figure}

    Each vertex has a base weight of $2\delta$, and certain vertices have an added weight of $\pm w$, for some $w > 0$. The highlighted vertices represent an independent set $S$, which corresponds to an invalid assignment and has $w_\mathrm{circ}(S, G_C) = 8\delta + w$. This is in fact the maximum circuit weight that can be attained by an independent set corresponding to an invalid assignment, and for each of the two valid assignments there is an independent set $S^*$ with~$w_\mathrm{circ}(S^*, G_C) =~9\delta$. Therefore, $G_C$ is a correct compilation of $C$ for $k = 9$ and $\tilde{w} = w$.

    On the other hand, consider the induced subcircuit graph $G_1$ of $T_1$ and the independent set $S_1 = S \cap V(G_1)$, which still corresponds to an invalid assignment for $T_1$. Notice that~$w_\mathrm{circ}(S_1, G_1) = 4\delta + w$, while each independent set $S_1^*$ corresponding to a valid assignment has circuit weight $w_\mathrm{circ}(S_1^*, G_1) = 5\delta$, so $k_1 = 5$ and $\tilde{w}_1 \ge w$. Likewise, we have~$k_2 = 5$ and $\tilde{w}_2 = w$. Theorem \ref{thm:correctness} only tells us that $G_C$ is a correct compilation of $C$ for~$\tilde{w} = \tilde{w}_1 + \tilde{w}_2 = 2w$, but we know it is still correct even for $\tilde{w} = w$.
\end{example}

\section{QAP to UDG-MWIS Reduction}\label{reduction}
We now show a reduction from \ac{QAP} to UDG-MWIS, enabling us to use Rydberg arrays to find solutions for \ac{QAP}. Specifically, given a \ac{QAP} instance $I = (F, D)$ we construct a weighted graph $G = (V, E, w)$ from which we can ``read off'' the optimal placement $\Pi$ for $I$ from the maximum weight independent set $S$ of $G$. 

\subsection{Initial Steps}
To introduce our reduction, we will first present a slightly more intuitive but less efficient one. This reduction closely follows the techniques from \cite{nguyen2023quantum} with circuit tiles introduced the previous section instead of gadgets. Rather than encoding our \ac{QAP} instance as a weighted graph then converting it via gadgets to an UDG, we reduce first to a circuit then compile, tile by tile, to an UDG:
\begin{enumerate}
    \item First, and most simply, we create a \textbf{crossing lattice} with a wire $(x,i)$ for each of the~$n^2$ binary variables $\pi_{xi}$, as shown in Figure \ref{fig:lattice} for a $2\times 2$ \ac{QAP} instance. We place the wires from left to right in lexicographic order, \textit{i.e.} $(x, i) < (y,j)$ if $x < y$ or $x = y$ and~$i < j$, arranged such that each wire intersects each other wire exactly once. The wires are arranged such that wherever two wires $(x,i) < (y,j)$ intersect, the horizontal wire is $(x,i)$ and the vertical wire $(y,j)$.
     
    \item Next, we consider the \textbf{constraints} on these variables: namely, no facility can be placed in two different locations and no two facilities can be placed in the same location. We encode this by placing a (1,1)-restriction at the intersection of each pair of wires $(x,i)$ and $(y,j)$ for which either $x = y$ or $i = j$, as shown in Figure \ref{fig:constraints}. The reader may notice we have omitted one constraint, that no facility can go unplaced nor any location unfilled. We will encode this constraint in the next step as an incentive for activating as many wires as possible, rather than as a penalty for activating too few.
    
    \item Having established our constraints, we now encode the \textbf{cost function}, as shown in Figure \ref{fig:weights}. At each intersection without a restriction (\textit{i.e.} the intersection of wires $(x,i)$ and $(y,j)$ where $x \ne y$ and $i \ne j$) we add a (1,1)-bias $w_{xi,yj} = 2w_0  - f_{xy} d_{ij} - f_{yx} d_{ji}$, where $w_0 = \max{\{f_{xy} d_{ij}\}_{x,y,i,j=1}^n} + \epsilon$ for some $\epsilon > 0$. On each individual wire $(x,i)$ we also add a 1-bias $w_{xi} = w_0  - f_{xx} d_{ii}$. Thus, the weight of a valid assignment $\Pi = (\pi_{xi})$ (not necessarily a permutation matrix) is
    \begin{align*}
        w(\Pi)
        &= \sum_{x,i=1}^n \pi_{xi} w_{xi} + \sum_{(x,i)<(y,j)} \pi_{xi} \pi_{yj} w_{xi,yj} \\
        &= \sum_{x,i=1}^n \pi_{xi} (w_0 - f_{xx} d_{ii}) + \sum_{(x,i)<(y,j)} \pi_{xi} \pi_{yj} (2w_0 - f_{xy} d_{ij} - f_{yx} d_{ji}) \\
        &= \sum_{x,y,i,j=1}^n \pi_{xi} \pi_{yj} (w_0 - f_{xy} d_{ij}) \\
        &= w_0\sum_{x,y,i,j=1}^n \pi_{xi} \pi_{yj} - C(\Pi).
    \end{align*}
    Since $w_0 > f_{xy} d_{ij}$ for all $(x,y,i,j)$, any assignment $\Pi$ that maximizes $w(\Pi)$ will have that $\pi_{xi} \pi_{yj} = 1$ for as many tuples $(x,y,i,j)$ as possible. From our constraints we know that, if we fix a ``row'' $x$ or ``column'' $i$, at most one wire $(x,i)$ can be activated. Therefore, the assignments optimizing $\sum_{x,y,i,j=1}^n \pi_{xi} \pi_{yj}$ are those in which for every $x$ or $i$ exactly one wire $(x, i)$ is activated; these are precisely the assignments which correspond to permutation matrices. For such an assignment, we have that $\sum_{x,y,i,j=1}^n \pi_{xi} \pi_{yj} = \frac{n(n+1)}{2}$ and thus~$w(\Pi) = \frac{n(n+1)}{2} w_0 - C(\Pi)$, so the optimal assignment for the circuit will be the one that minimizes $C(\Pi)$.
  
    \item Finally, we \textbf{compile} the circuit tile by tile to produce the weighted unit disk graph $G$ in Figure \ref{fig:graph}. For properly normalized weights, the compilation procedure guarantees that the MWIS of $G$ corresponds to the optimal assignment for the circuit. We can read off the optimal assignment $\Pi$ from the four vertices at the top of the graph:~$\pi_{xi} = 1$ if and only if the corresponding vertex is included in the MWIS.
\end{enumerate}

\begin{figure}[ht]
    \centering
    \begin{subfigure}[t]{0.4\linewidth}
        \centering
        \includegraphics[width=0.7\linewidth]{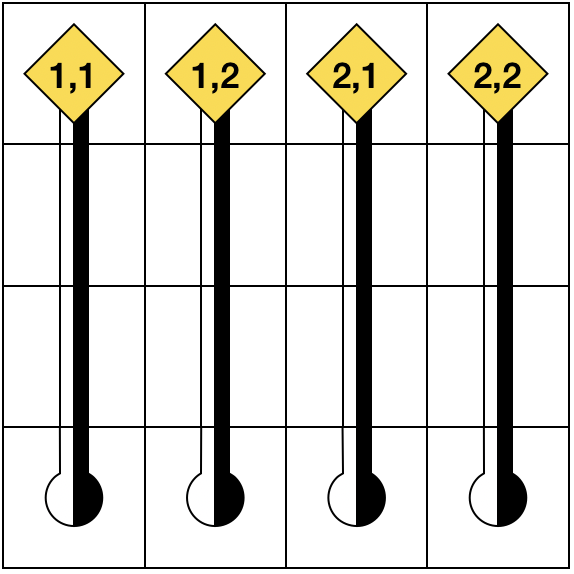}
        \caption{Create wires for variables}
        \label{fig:lattice}
    \end{subfigure}
    \begin{subfigure}[t]{0.4\linewidth}
        \centering
        \includegraphics[width=0.7\linewidth]{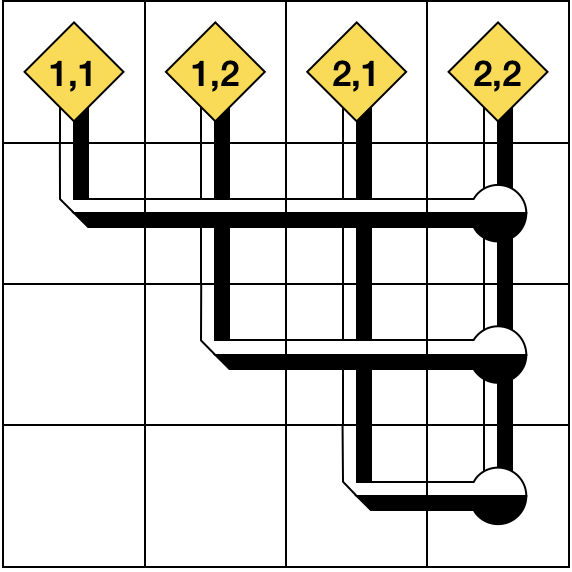}
        \caption{Arrange into crossing lattice}
        \label{fig:constraints}
    \end{subfigure}
    \par\bigskip
    \begin{subfigure}[t]{0.4\linewidth}
        \centering
        \includegraphics[width=0.7\linewidth]{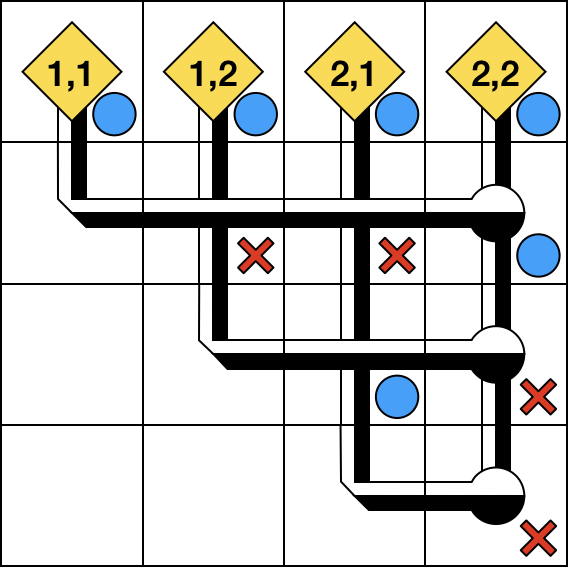}
        \caption{Add biases and constraints}
        \label{fig:weights}
    \end{subfigure}
    \begin{subfigure}[t]{0.4\linewidth}
        \centering
        \includegraphics[width=0.7\linewidth]{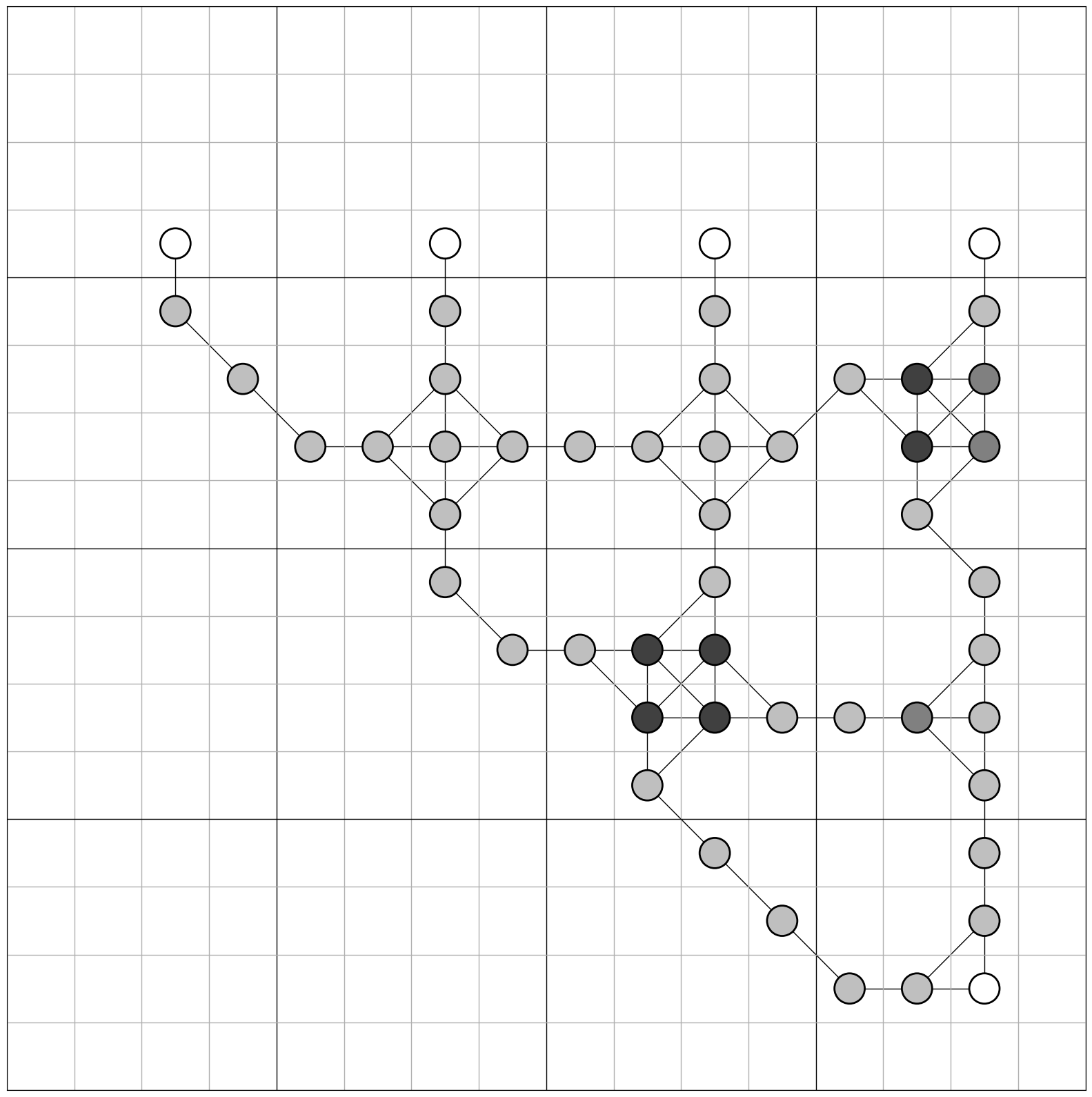}
        \caption{Compile to graph}
        \label{fig:graph}
    \end{subfigure}
    \caption{Transferring the crossing lattice to a graph}
\end{figure}

\subsection{Reduction}
We begin by choosing how to formulation a \ac{QAP} instance as a constrained binary optimization problem. Recall our 'canonical' formulation from Section \ref{sec:tiles-cbop}:
\begin{enumerate}
    \item \textbf{Variables:} $\pi_{xi}$ for each $x, i \in [n]$
    \item \textbf{Constraints:} $C_{x*} : \sum_{i=1}^n \pi_{xi} = 1$ and $C_{*i} : \sum_{x=1}^n \pi_{xi} = 1$ for each $x, i \in [n]$
    \item \textbf{Weight Function:}
    $w(\Pi) = -\sum_{x,y,i,j=1}^n f_{xy} d_{ij} \pi_{xi} \pi_{yj}$
\end{enumerate}
If an assignment $\Pi$ satisfies constraint $C_{x*}$, we can write $\pi_{xn} = 1 - \sum_{i=1}^{n-1} \pi_{xi}$. Thus, we do not need to explicitly include $\pi_{xn}$ as a variable; we can just define it implicitly using~$\pi_{x1}, \dots \pi_{x(n-1)}$. Since the implied $\pi_{xn}$ must still be a binary variable, in the modified problem we replace constraint $C_{x*}$ with $C_{xn}' : \sum_{i=1}^{n-1} \pi_{xi} \le 1$. We can modify each row $s$ of~$\Pi$ in this way, reducing the number of variables from $n \times n$ to $n \times (n - 1)$. We can do likewise for each column $i$: write $\pi_{ni} = 1 - \sum_{x=1}^{n-1} \pi_{xi}$ and replace $C_{*i}$ with $C_{ni}' : \sum_{x=1}^{n-1} \pi_{xi} \le 1$. Some care is required for column $n$; since we have already removed every variable $\pi_{xn}$ we define $\pi_{nn}$ as follows:
$$\pi_{nn} = 1 - \sum_{x=1}^{n-1} \pi_{xn} = 1 - \sum_{x=1}^{n-1} \left(1 - \sum_{i=1}^{n-1} \pi_{xi} \right) = 2 - n + \sum_{x,i=1}^{n-1} \pi_{x,i}.$$

We also replace constraint $C_{*n}$ with $C_{*n}' : \pi_{nn} \in \{0, 1\}$, that is $n - 2 \le \sum_{x,i=1}^{n-1} \pi_{x,i} \le n - 1$ (if constraints $C_{*1}', \dots, C_{*(n-1)}'$ are already satisfied, the right-hand inequality is trivially met). This is in fact identical to $C_{n*}'$, therefore, we can replace both $C_{n*}$ and $C_{*n}$ with a single constraint $C_{nn}' : \sum_{x,i=1}^{n-1} \pi_{xi} \ge n - 2$. As a consequence, we have a new set of variables~$\Pi|_{n-1} = (\pi_{xi})_{x,i=1}^{n-1}$, with the following equations to define the remaining $2n - 1$ variables:
\begin{align}
\label{eqn:vars1}
    \pi_{xn} &= 1 - \sum_{i=1}^{n-1} \pi_{xi}, \forall x \in [n-1] \\
\label{eqn:vars2}
    \pi_{ni} &= 1 - \sum_{x=1}^{n-1} \pi_{xi}, \forall i \in [n-1] \\
\label{eqn:vars3}
    \pi_{nn} &= 2 - n + \sum_{x,i=1}^{n-1} \pi_{xi}.
\end{align}
Closely related, we have a new set of $2n - 1$ constraints enforcing that these implicit variables are indeed either $0$ or $1$. We now turn our attention to the weight function $w(\Pi|_{n-1})$. In order to encode \ac{QAP}, we should have $w(\Pi|_{n-1}) = w_0 - C(\Pi)$ for every valid assignment~$\Pi$. For ease of representation as a weighted circuit, we want to express $w$ as a quadratic polynomial of the form
$$w(\Pi|_{n-1}) = \sum_{\substack{x,y,i,j=1 \\ x<y}}^{n-1} w_{xi,yj} \pi_{xi} \pi_{yj} + \sum_{x,i=1}^{n-1} w_{xi} \pi_{xi}.$$
To compute the coefficients $w_{xi}$ and $w_{xi,yj}$ we first write out the \ac{QAP} cost function as a polynomial in all $n^2$ variables, $C(\Pi) = \sum_{x,y,i,j=1}^n f_{xy} d_{ij \pi_{xi} \pi_{yj}}$. We then use Equations \ref{eqn:vars1}, \ref{eqn:vars2}, and \ref{eqn:vars3} to write the implicit variables $\pi_{xn}$, $\pi_{ni}$, and $\pi_{nn}$ in terms of $\Pi|_{n-1}$. After gathering terms we have the following coefficients, where $f_{xy}' = f_{xy} - f_{xn} - f_{ny} + f_{nn}$ and~$d_{ij}' = d_{ij} - d_{in} - d_{nj} + d_{nn}$:
\begin{align*}
    w_{xi}
    &= -2f_{xx}' d_{ii}' + \sum_{m=1}^n (f_{xm}' d_{im}'
        - (f_{xm} - f_{xn}) (d_{im} - d_{nm}) \\
        &\qquad\qquad\qquad\qquad + f_{mx}' d_{mi}'
        - (f_{mx} - f_{nx}) (d_{mi} - d_{mn})) \\
    & \quad \\
    w_{xi,yj}
    &= -(f_{xy}' d_{yj}' + f_{yx}' d_{ji}')
\end{align*}

Finally, we can fully describe the \ac{QAP} as a constrained binary optimization problem in~$(n-1)^2$ variables:
\begin{enumerate}
    \item \textbf{Variables:} $\pi_{xi}$ for each $x, i \in [n - 1]$
    \item \textbf{Constraints:}
    \begin{align*}
        C_{xn}' &: \sum_{i=1}^{n-1} \pi_{xi} \le 1, \forall x \in [n-1] \\
        C_{ni}' &: \sum_{x=1}^{n-1} \pi_{xi} \le 1, \forall i \in [n-1] \\
        C_{nn}' &: \sum_{x,i=1}^{n-1} \pi_{xi} \ge n-2
    \end{align*}
    \item \textbf{Weight Function:} with $w_{xy,ij}$ and $w_{xi}$ as defined above,
    $$w(\Pi|_{n-1}) = \sum_{\substack{x,y,i,j=1 \\ x<y}}^{n-1} w_{xi,yj} \pi_{xi} \pi_{yj} + \sum_{x,i=1}^{n-1} w_{xi} \pi_{xi}.$$
\end{enumerate}

We now construct a weighted circuit for this formulation of \ac{QAP}. We start by defining two useful subcircuits. The chain $C_\mathsf{OR}(k)$ of (1,1)-restricted OR gates in Figure \ref{fig:or-chain} ensures that no two of $x_1, \dots, x_k$ both have the value $1$ and that $x_0 = x_1 \vee \cdots \vee x_k = x_1 + \cdots + x_k$. This will be useful in enforcing the row constraints $C_{xn}'$. The chain $C_\mathsf{AND}(k)$ of (0,0)-restricted AND gates in Figure \ref{fig:and-chain} ensures that no two of $x_1, \dots, x_k$ both have value $0$; that is, $x_1 + \cdots + x_k \ge k - 1$. This will be useful in combining the sums of rows to enforce the constraint $C_{nn}'$.

\begin{figure}[h!]
    \centering
    \begin{subfigure}[t]{0.7\linewidth}
        \centering
        \includegraphics[width=0.7\linewidth]{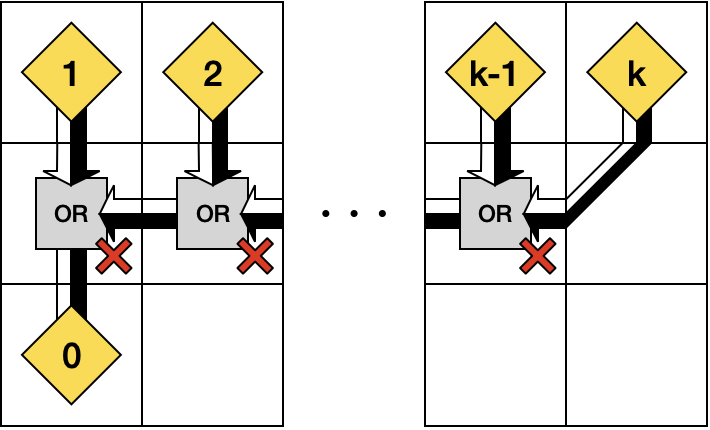}
        \caption{Circuit $C_\mathsf{OR}(k)$ to enforce $x_0 = x_1 + \cdots + x_k$}
        \label{fig:or-chain}
    \end{subfigure}
    \par\bigskip
    \begin{subfigure}[t]{0.7\linewidth}
        \centering
        \includegraphics[width=0.7\linewidth]{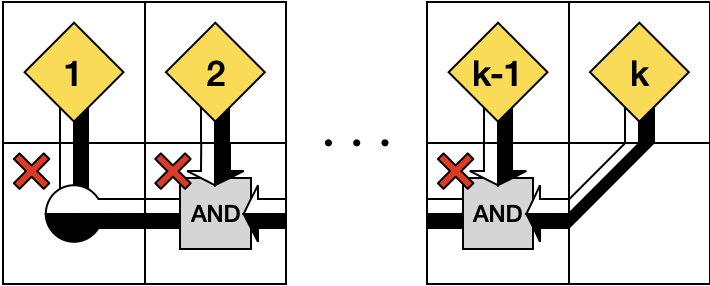}
        \caption{Circuit $C_\mathsf{AND}(k)$ to enforce $1 + x_1 + \cdots + x_k \ge k$}
        \label{fig:and-chain}
    \end{subfigure}
    \caption{Two subcircuits to be used in encoding QAP}
\end{figure}

To build our circuit, we create a wire $(x, i)$ for each variable $\pi_{xi}$ of $\Pi|_{n-1}$. We then arrange these wires in a modified crossing lattice: each pair of wires $(x, i)$ and $(x,j)$ for the same facility $x$ are parallel, while each pair of wires $(x, i)$ and $(y, j)$ with $x \ne y$ intersect exactly once. When $i = j$ we apply a $(1,1)$-restriction at this intersection to help enforce constraint $C_{ni}'$; when $i \ne j$ we apply a $(1,1)$-bias with weight $w_{xi,y_j}$. To each individual wire $(x, i)$ we apply a $1$-bias with weight $w_{xi}$. Finally, we create $n-1$ copies $C_1, \dots, C_{n-1}$ of $C_\mathsf{OR}(n-1)$ and, for each $x \in [n - 1]$, connect the ``input'' wires $1, \dots, n-1$ of $C_x$ to the wires $(x, i), \dots, (x, n-1)$ and its ``output'' wire $0$ to wire $x$ of $C_\mathsf{AND}(n-1)$. The resulting circuit for a $4 \times 4$ \ac{QAP} instance is shown in Figure \ref{fig:4x4-tiles}, with its compilation to UDG-MWIS in Figure \ref{fig:4x4-graph}.

\begin{theorem} \label{thm:qap-reduction}
Given a \ac{QAP} instance $I = (F, D)$, with circuit $C_I$ and graph $G_I$ constructed as above, if
$\delta > \max_{x,i}\{|w_{xi}| + \sum_{y,j}|w_{xi,yj}|\}$ then the maximum-weight independent set(s) of $G_I$ corresponds to the optimal solution(s) of $I$.
\end{theorem}

\begin{proof}
We have shown how to encode $I$ as the following constrained binary optimization problem $P_I$.
\begin{enumerate}
    \item \textbf{Variables:} $\pi_{xi}$ for each $x, i \in [n - 1]$
    \item \textbf{Constraints:}
    \begin{align*}
        C_{xn}' &: \sum_{i=1}^{n-1} \pi_{xi} \le 1, \forall x \in [n-1] \\
        C_{ni}' &: \sum_{x=1}^{n-1} \pi_{xi} \le 1, \forall i \in [n-1] \\
        C_{nn}' &: \sum_{x,i=1}^{n-1} \pi_{xi} \ge n-2
    \end{align*}
    \item \textbf{Weight Function:}
    $$w(\Pi|_{n-1}) = \sum_{\substack{x,y,i,j=1 \\ x<y}}^{n-1} w_{xi,yj} \pi_{xi} \pi_{yj} + \sum_{x,i=1}^{n-1} w_{xi} \pi_{xi}$$
\end{enumerate}
To see that the weighted circuit $C_I$ encodes $P_I$, observe the following:
\begin{itemize}
    \item For a given location $i$, the constraint $C_{ni}'$ is satisfied if and only if $\pi_{xi} \pi_{yi} = 0$ for each pair of distinct facilities $x$ and $y$. This is precisely what the $(1,1)$-restrictions at each intersection of the wires $(x,i)$ and $(y, i)$ enforce.
    \item For a given facility $x$, the constraint $C_{xn}'$ that $\pi_{x,1} + \dots + \pi_{x,n-1} \le 1$ is enforced by the subcircuit $C_x$, since the value of its output wire must be either $0$ or $1$.
    \item Since the output wire of each subcircuit $C_x$ is connected to the wire $x$ of $C_\mathsf{AND}(n-1)$, the subcircuit $C_\mathsf{AND}(n-1)$ enforces $C_{nn}'$:
    $$1 + \sum_{x=1}^{n-1} \left( \sum_{i=1}^{n-1} \pi_{xi} \right) \ge n - 1.$$
    \item Each intersection weight $w_{xi,yj}$ is applied precisely when $\pi_{xy}\pi_{yj} = 1$, while each wire weight $\pi_{xi}$ is applied precisely when $\pi_{xi} = 1$. Therefore the weight in $C$ of a valid assignment $\pi$ is identical to its weight in $P$:
    $$w(\pi) = \sum_{\substack{x,y,i,j=1 \\ x<y}}^{n-1} w_{xi,yj} \pi_{xi} \pi_{yj} + \sum_{x,i=1}^{n-1} w_{xi} \pi_{xi}.$$
\end{itemize}
Next, we need to show that $G_I$ is a correct compilation of $C_I$. Since the graph compilation $G_T$ of each circuit tile $T$ has no more than 16 vertices (in fact, none of our tiles have more than~8), one can establish its correctness simply by computing the weight of each independent set. While our graph has a few irregularities---the graphs for some of the diagonal wires have a connecting vertex in a corner, and the $(1,1)$-restricted OR gates on the bottom row have an apparent connecting vertex on their bottom edge---it is easy to see that these do not affect the correctness of the entire graph. We can therefore apply Theorem \ref{thm:correctness} to conclude that~$G_I$ is a correct compilation of $C_I$.

Finally, observe that we state a condition $\delta > \max_{x,i}\{|w_{xi}| + \sum_{y,j}|w_{xi,yj}|\}$ less restrictive the condition $\delta > \tilde{w}$ in Theorem \ref{thm:correctness}. We are able to improve on that very conservative bound by considering the structure of our circuit $C$. In particular, for each tile $T_i$ in our circuit (before applying biases), each independent set of $G_i$ with weight $(k_i - t)\delta$ corresponds to a valid assignment for that tile with the wire value on at most $t$ edges flipped. Moreover, each assignment to one of our restricted logic gates which results from flipping the output wire of a valid assignment can alternatively by produced by flipping one of the input wires of a (possibly different) valid assignment to that tile. By extension, any wire value flipped in the logic section of $C$ can be modeled by flipping a wire value in the crossing lattice section. Therefore, each independent set of $G_I$ with (pre-bias) weight $(k - t)\delta$ (``with $t$ errors'') corresponds to a valid assignment with its wire values flipped at no more than $t$ edges in the crossing lattice. Let $\pi$ be a (possibly invalid) assignment to the wires of $C$, and~$\pi'$ an assignment produced by flipping wire $(x, i)$ at one edge in the crossing lattice; then~$w(\pi') \le w(\pi) + |w_{xi}| + \sum_{y,j}|w_{xi,yj}|$. Let $\pi^*$ be a maximum-weight valid assignment to~$C$ with corresponding independent set $S^*$, and $S'$ an independent set of $G_I$ corresponding to an invalid assignment $\pi^*$ with $t$ errors. Then, $w(\pi') \le w(\pi) + t \cdot \max_{x,i}\{|w_{xi}| + \sum_{y,j}|w_{xi,yj}|\}$ for some valid assignment $\pi$, meaning
\begin{align*}
    w(S')
    &\le (k - t)\delta + w(\pi)    
        + t \cdot \max_{x,i}\{|w_{xi}| + \sum_{y,j}|w_{xi,yj}|\} \\
    &< (k - t)\delta + w(\pi) + t \delta \\
    &\le k\delta + w(\pi^*) \\
    &= w(S^*).
\end{align*}
Therefore, the maximum-weight independent set of $G_I$ corresponds to a maximum-weight valid assignment for $C_I$, and thus to an optimal solution for $I$.
\end{proof}

\begin{figure}[h!]
    \centering
    \begin{subfigure}[t]{0.5\linewidth}
        \centering
        \includegraphics[width=\linewidth]{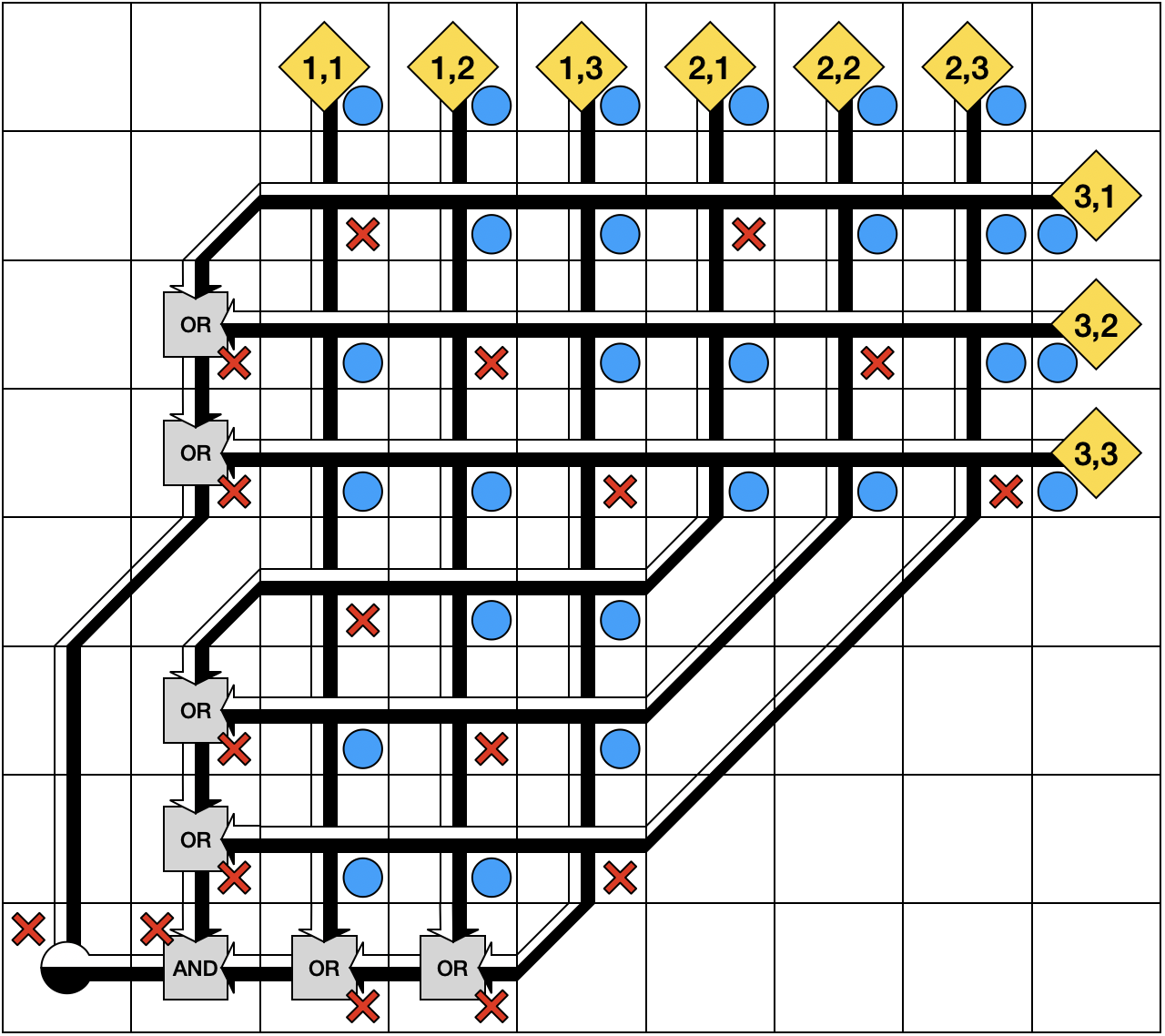}
    
        \caption{Weighted circuit encoding $4 \times 4$ QAP}
        \label{fig:4x4-tiles}
    \end{subfigure}
    \par\bigskip
    \begin{subfigure}[t]{0.5\linewidth}
        \centering
        \includegraphics[width=\linewidth]{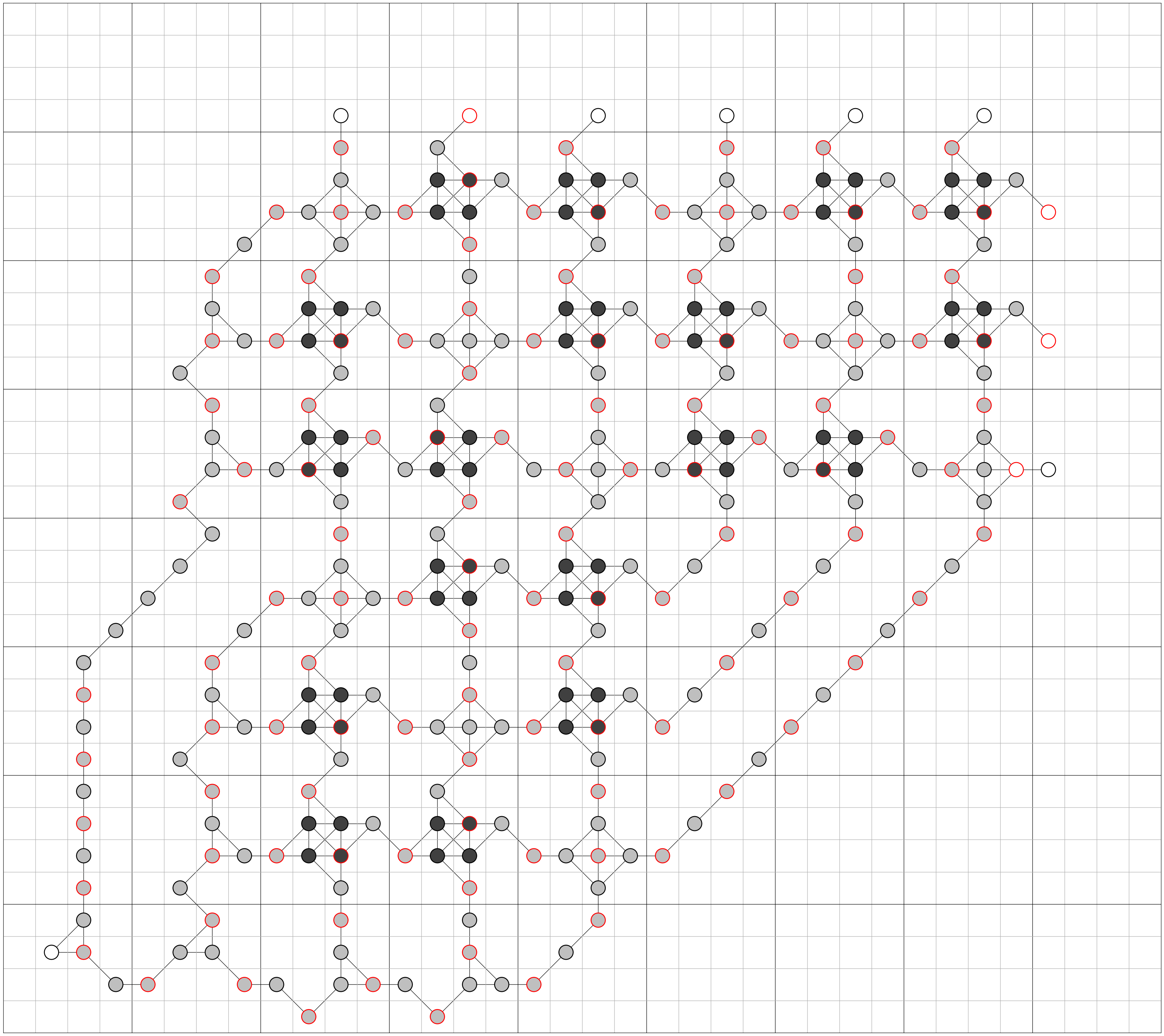}
        
        \caption{UDG-MWIS instance encoding $4 \times 4$ QAP}
        \label{fig:4x4-graph}
    \end{subfigure}
    \captionsetup{width=0.85\linewidth}
    \caption{Encoding of the $4 \times 4$ QAP, shown with a maximum-weight independent set corresponding to the assignment $(\pi_1, \pi_2, \pi_3, \pi_4) = (2, 4, 3, 1)$}
\end{figure}

\section{Conclusion and Future Work}

In this paper, we provided an algorithm to find valid and optimal solutions to the \ac{QAP}. We extended the techniques used for solving the \ac{MWIS} problem to \ac{QAP} by providing a reduction from \ac{QAP} to UDG-MWIS. This can be exploited to determine the placement relationships of Rydberg atoms that, when excited, will find solutions for \ac{QAP}. An optimized circuit algorithm was developed for \ac{QAP}, significantly reducing the number of atoms required for problem representation and improving overall efficiency. The visual language provided describes a broader class of constrained binary optimization problems and applying this language to \ac{QAP} provides valuable insight into encoding other complex problems using Rydberg arrays. Additionally, it facilitates the design of quantum algorithms for a broader range of computational challenges, enabling the exploration of novel applications in quantum computing. The Aquila machine from QuEra is a natural UDG-MIS solver, and the formulation provided in this paper can be used to solve QAP on it.

It remains open to test the algorithm on small \ac{QAP} problems on quantum hardware and investigate the probability of valid solutions and optimal solutions. The ability to local detune the excitation field as required to run the algorithm was released recently and requires expert knowledge in quantum to be integrated properly. One could analyze the scaling in terms of the number of atoms needed to encode an $n\times n$ \ac{QAP} problem, execution time and success probability. One could also use the visual language and reduction provided in this paper and apply it to other problems that can be formulated using Rydberg arrays.

\section*{Acknowledgements}
The authors would like to thank Giuseppe Cotardo, Jason LeGrow and Gretchen Matthews for comments on previous drafts of this paper, as well as Charlotte Lowdermilk and Ehren Hill for their continued support during this project.

\bibliographystyle{plain}
\bibliography{QuantumQAP}

\end{document}

%% file: acronyms.tex
\DeclareAcronym{QUBO}{
	short = QUBO,
	long = Quadratic Unconstrained Binary Optimization,
}

\DeclareAcronym{QAP}{
	short = QAP,
	long = Quadratic Assignment Problem,
}

\DeclareAcronym{MIS}{
	short = MIS,
	long = Maximum Independent Set,
}

\DeclareAcronym{MWIS}{
	short = MWIS,
	long = Maximum Weighted Independent Set,
}

\DeclareAcronym{UDG}{
	short = UDG,
	long = Unit Disk Graph,
}